\documentclass[journal, twoside, final]{IEEEtran}
\usepackage{cite}
\usepackage{graphicx}
\usepackage{amsmath}
\usepackage{amssymb}
\usepackage{amsthm}
\usepackage{amsfonts}
\usepackage{algorithm}
\usepackage[noend]{algpseudocode}
\usepackage{array}
\usepackage{color}
\usepackage{enumerate}
\usepackage{subfigure}
\usepackage{url}
\usepackage{placeins}

\makeatletter

\usepackage{algorithm, algpseudocode, float}
\usepackage{lipsum}

\makeatletter
\newenvironment{breakablealgorithm}
  {% \begin{breakablealgorithm}
   \begin{center}
     \refstepcounter{algorithm}% New algorithm
     \hrule height.8pt depth0pt \kern2pt% \@fs@pre for \@fs@ruled
     \renewcommand{\caption}[2][\relax]{% Make a new \caption
       {\raggedright\textbf{\ALG@name~\thealgorithm} ##2\par}%
       \ifx\relax##1\relax % #1 is \relax
         \addcontentsline{loa}{algorithm}{\protect\numberline{\thealgorithm}##2}%
       \else % #1 is not \relax
         \addcontentsline{loa}{algorithm}{\protect\numberline{\thealgorithm}##1}%
       \fi
       \kern2pt\hrule\kern2pt
     }
  }{% \end{breakablealgorithm}
     \kern2pt\hrule\relax% \@fs@post for \@fs@ruled
   \end{center}
  }
\makeatother

\begin{document}

\title{\textcolor{black}{Caching} at Base Stations with Multi-Cluster Multicast
\textcolor{black}{Wireless} Backhaul via Accelerated First-Order Algorithm}

\author{Yang~Li,~\IEEEmembership{Student Member,~IEEE,}
        Minghua~Xia,~\IEEEmembership{Member,~IEEE,}
        and~Yik-Chung~Wu,~\IEEEmembership{Senior Member,~IEEE}
\thanks{Manuscript received March 7, 2019; revised August 20, 2019 and October 29, 2019; accepted January 18, 2020.
This work was supported in part by the National Natural Science Foundation of China under Grant 61671488,  in part by the Major Science and Technology Special Project of Guangdong Province under Grant 2018B010114001, and in part by the Fundamental Research Funds for the Central Universities under Grant 191gjc04.
The associate editor coordinating the review of this paper and approving it for publication was A. Abrardo. }
\thanks{Y. Li and Y.-C. Wu are
with the Department of Electrical and Electronic Engineering, The
University of Hong Kong, Hong Kong (e-mail:
\{liyang, ycwu\}@eee.hku.hk).}
\thanks{M. Xia is with the School of Electronics and Information Technology, Sun Yat-sen University, Guangzhou, 510006, China, and with Southern Marine Science and Engineering Guangdong Laboratory (Zhuhai) (e-mail: xiamingh@mail.sysu.edu.cn).}
\thanks{Color versions of one or more of the figures in this paper are available online at http://ieeexplore.ieee.org.}
\thanks{Digital Object Identifier}
}

\IEEEpubid{\begin{minipage}{\textwidth} \ \\[12pt] \centering 1536-1276 \copyright\ 2018 IEEE. Translations and content mining are permitted for academic research only. Personal use is also permitted, \\
but republication/redistribution requires IEEE permission. See \url{http://www.ieee.org/publications_standards/publications/rights/index.html for more information}.\end{minipage}}

\maketitle
\begin{abstract}
Cloud radio access network (C-RAN) has been recognized as a promising architecture for next-generation wireless systems to \textcolor{black}{support} the rapidly increasing demand for higher data rate.
However, the performance of C-RAN is limited by the backhaul capacities, especially for the wireless deployment. While C-RAN with fixed BS caching has been demonstrated to reduce backhaul consumption,
it is more challenging to further optimize the cache allocation at BSs with multi-cluster multicast backhaul, where the inter-cluster interference induces additional non-convexity to the cache optimization problem.
Despite the challenges, we propose an accelerated first-order algorithm,
which achieves much higher content downloading sum-rate than a second-order algorithm running for the same amount of time. Simulation results demonstrate that,  by simultaneously delivering the required contents to different multicast clusters,
the proposed algorithm achieves significantly higher downloading sum-rate than those of
time-division single-cluster transmission schemes. Moreover, it is found
that the proposed algorithm allocates larger cache sizes to the farther BSs within the nearer clusters,
which provides insight to the superiority of the proposed cache allocation.
\end{abstract}


\section{Introduction}
To meet the dramatically increasing demand for higher data rate, cloud radio access network (C-RAN), where the base stations (BSs) are connected to a computation center via high-speed backhaul links for multi-BS cooperation, is a promising architecture for next-generation wireless systems \cite{Rost2014,Simeone2016,Quek2017}. However, the performance of C-RAN is mainly limited by the backhaul
capacities from the computation center to the BSs,
especially for the small-cell deployment, where high-speed optical fiber connections may not be available \cite{Dai2016}, and wireless backhaul is the only option.

On the other hand, with modern wireless data traffic being more and more dominated by
videos and other multimedia data, content-centric communications
exploiting multicast transmission and BS caching
draw a lot of attention lately \cite{Lecompte2012,Golrezaei2013,LiuH2014}.
As multiple BSs in the same cluster share the same users' data for BS cooperation, by multicasting users' messages from the computation center to these BSs simultaneously, the broadcast nature of the wireless backhaul channels can be efficiently exploited. Furthermore, by proactively caching a fraction of
popular contents at each BS, the amount of data to be delivered
through the wireless backhaul is reduced, thus improving the system
efficiency in terms of content downloading rate \cite{Dai2018}.

While C-RAN with BS caching has been investigated in
\cite{Ugur2016,Tao2016,Li2018,Park2016}, they all assume fixed cache allocation among BSs and focus on how BS caching helps to improve the system performance. In particular, \cite{Ugur2016,Tao2016,Li2018} investigate how BS caching facilitates the reduction of backhaul burden and power consumption. In \cite{Park2016}, data-sharing and compression are combined to examine how BS caching help in improving the spectral efficiency. On the other hand, although cache optimization has been studied in
\cite{Gitzenis2013,Shanmugam2013,Bastug2015,Cui2016,Xu2017}, they only
focus on the layer between the BS and the users, without considering the limitation of
the backhaul efficiency.
To the best of our knowledge, only the pioneering work \cite{Dai2018} investigates cache optimization at BSs aiming at improving the backhaul efficiency between the computation center and the BSs.

Unfortunately, since \cite{Dai2018} only considered a simplified C-RAN setup with a single cluster of BSs, the resulting caching scheme could not directly generalize to the more practical scenario with multiple BS clusters. For the multi-cluster scenario, the computation center is required to transmit different multicast data to different BS clusters simultaneously, resulting in inter-cluster interference, which in turn induces additional non-convexity to the cache optimization problem.

\IEEEpubidadjcol

Despite the challenges mentioned above, this paper optimizes the cache allocation at BSs
for a C-RAN with multi-cluster multicast backhaul,
aiming to maximize the content downloading sum-rate of the wireless backhaul
under a total cache budget constraint.
Since cache placement impacts a much larger timescale than that of channel variations \cite{Park2013,Patil2015,Dai2016}, the cache allocation
should be optimized based on a large number of potential channel realizations.
Furthermore, to maximize the content downloading sum-rate,
various channel realizations requires tailored optimal beamformers,
which are coupled in the optimization of cache sizes.
Consequently, with a large number of beamfomers being nuisance variables,
the cache allocation is a large-scale nonsmooth nonconvex problem. To solve this problem, we first tackle the non-smoothness and non-convexity by introducing auxiliary variables and constructing a sequence of quadratic convex functions in the successive convex approximation (SCA) framework. But instead of directly solving each convexified problem with the interior-point method, we further construct a strongly convex upper bound of the cost function, so that an accelerated first-order algorithm can be developed for solving each SCA subproblem in its dual domain.

Simulation results show that the proposed accelerated first-order algorithm
achieves much higher content downloading sum-rate than a second-order algorithm running for the same amount of time. Moreover,
by simultaneously delivering the required contents to different multicast clusters,
the proposed algorithm achieves significantly higher downloading sum-rate than those of time-division single-cluster transmission schemes. Finally, it is found that the proposed algorithm proactively allocates larger cache sizes to the farther BSs within the nearer clusters, which provides insight to the superiority of the proposed cache allocation.

The remainder of this paper is organized as follows. System model and problem formulation are introduced in Section II. In Section III, an accelerated first-order algorithm is proposed for the cache allocation. The multi-cluster multicast beamforming (MCMB) design for content delivery is presented in Section IV. Simulation results and discussions are provided in Section V. Section VI concludes the paper.

Throughout this paper, scalars, vectors, and matrices are denoted by lower-case letters (e.g., $a$), lowercase bold letters (e.g., $\mathbf{a}$), and upper bold letters (e.g., $\mathbf{A}$), respectively. The complex domain is denoted by $\mathbb{C}$. We denote the transpose and conjugate transpose of a vector/matrix by $(\cdot)^T$ and $(\cdot)^H$, respectively.
The real part, trace, and Frobenius-norm of a matrix are denoted by
$\Re(\cdot)$, $\mathrm{Tr}(\cdot)$, and $\Vert\cdot\Vert_F$, respectively.
The expectation of a random variable is denoted by $\mathbb{E}[\cdot]$,
and the complex Gaussian distribution is represented as $\mathcal{CN}\left(\cdot, \cdot\right)$.

\section{System Model and Problem Formulation}
Consider a downlink C-RAN with $G$ clusters of BSs connected to a computation center through wireless backhaul. To effectively utilize the wireless medium, the computation center adopts multicast beamforming to deliver users' intended messages to each cluster,
and the BSs in each cluster serve their users through cooperative transmission with
data sharing \cite{Simeone2009,Dai2014}. Furthermore, to alleviate the backhaul burden, each BS is equipped with a local cache to pre-store a subset of popular files.
An example of such a cache enabled downlink C-RAN
is illustrated in Fig. \ref{system model}, where the BSs are clustered into $G$
disjoint clusters \cite{3GPP,Huang2009}.
In this paper, we assume that the BSs have been clustered,
and focus on how to optimally allocate the cache sizes among the BSs to improve the backhaul efficiency.

\begin{figure}[t!]
\centering
  \includegraphics[width=0.4\textwidth]{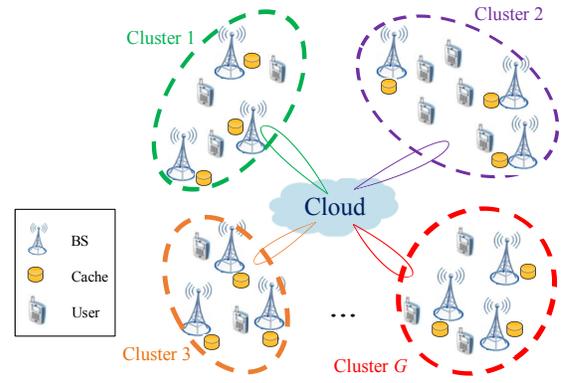}
  \caption{A downlink C-RAN consists of $G$ clusters of BSs, where each BS is equipped with a local cache.}
  \label{system model}
\vspace{-0.2cm}
\end{figure}

Let $M$ and $N$ ($M>N$) denote the numbers of antennas at the computation center and each BS, respectively.
Then the maximum number of independent data streams of each cluster is $d=\min\{M, N\}=N$,
and the multicast beamforming matrix from the computation centre
to the $g$-th cluster of BSs is denoted as $\mathbf{V}_g\in \mathbb{C}^{M\times d}$.
Denoting the total number of BSs as $K$, we can express the
received signal at the $k$-th BS as
\begin{eqnarray}\label{rs}
\mathbf{y}_k = \underbrace{\mathbf{H}_{k}\mathbf{V}_{g_k}\mathbf{x}_{g_k}}_{\text{desired multicast signal}}
+\underbrace{\sum_{g^\prime\neq g_k}\mathbf{H}_{k}\mathbf{V}_{g^\prime}\mathbf{x}_{g^\prime}}_{\text{inter-cluster interference}}+\mathbf{n}_k,
\end{eqnarray}
where $\mathbf{H}_k \in \mathbb{C}^{N\times M}$ is the channel matrix from the computation centre to BS $k$, $g_k\in\{1,2,\ldots,G\}$ is the index of the group to which BS $k$ belongs,
$\mathbf{x}_{g_k} \in \mathbb{C}^{d\times 1}$ is the data vector sent to
cluster $g_k$, and $\mathbf{n}_k\sim\mathcal{CN}\left(\mathbf{0}, \sigma_k^2\mathbf{I}_N\right)$ is the additive white Gaussian noise.
Based on \eqref{rs}, the mutual information between the
transmit signal $\mathbf{V}_{g_k}\mathbf{x}_{g_k}$ and
the received signal $\mathbf{y}_k$ can be written as
\begin{eqnarray}\label{mi}
I\left(\mathbf{V}_{g_k}\mathbf{x}_{g_k}; \mathbf{y}_k\right)=
\log\det\left(\mathbf{I}_N
+\mathbf{H}_k\mathbf{V}_{g_k}\mathbf{V}_{g_k}^H\mathbf{H}_k^H\mathbf{J}_k\right),
\end{eqnarray}
where $\mathbf{J}_k\triangleq\left(\sum_{g^\prime\neq g_k}\mathbf{H}_k\mathbf{V}_{g^\prime}\mathbf{V}_{g^\prime}^H\mathbf{H}_k^H
+\sigma_k^2\mathbf{I}_{N}\right)^{-1}$.

A central issue in C-RAN is to alleviate the backhaul burden during the
peak traffic time \cite{Dai2014}.
To address this issue, caching highly popular files at BSs during off-peak hours provides a viable solution \cite{Tao2016}\cite{Liu2016}.
However, the network operator has a fixed budget to deploy only a
limited amount of total cache size. Due to the limited
cache size,
\textcolor{black}{
each BS pre-stores fractions of popular contents during off-peak hours,
}and requests the rest from the computation center via wireless backhaul~\cite{Dai2018}\cite{Bidokhti2018}.
\textcolor{black}{
Specifically,
BS $k$ caches the first $C_k$ bits of the file requested by cluster $g_k$.
Denote $F_{g_k}$ as the total size of the file requested by cluster $g_k$,
then BS $k$ requires to receive the rest $F_{g_k}-C_k$ bits of the file from the computation center
when the file is delivered to mobile users \cite{Dai2018}\cite{Bidokhti2018}.}
With the knowledge of cached content at the BSs, an efficient joint cache-channel coding strategy \cite{Bidokhti2018} results in
the content downloading rate of cluster $g$ as
$R_g=\min_{k\in\mathcal{K}_g}
\left\{\frac{I\left(\mathbf{V}_{g}\mathbf{x}_{g}; \mathbf{y}_k\right)}{1-C_k/F_g}\right\}$ \cite{Dai2018},
where $\mathcal{K}_g$ denotes the set of BSs in cluster $g$.
By substituting the mutual information $I\left(\mathbf{V}_{g}\mathbf{x}_{g}; \mathbf{y}_k\right)$ into $R_g$, the downloading sum-rate of all the $G$ clusters of BSs can be written as
\begin{eqnarray}\label{DRsum}
R_{\text{sum}}
&=&\sum_{g=1}^G\min_{k\in\mathcal{K}_g}
\Bigg\{\frac{F_g}{F_g-C_k}\nonumber\\
&&
\cdot\log\det\left(\mathbf{I}_N
+\mathbf{H}_k\mathbf{V}_{g}\mathbf{V}_{g}^H\mathbf{H}_k^H
\mathbf{J}_k\right)\Bigg\}.~~~~~
\end{eqnarray}

To improve the backhaul efficiency, the cache sizes $\left\{C_k\right\}$ should be allocated to maximize the content downloading sum-rate in \eqref{DRsum}.
However, since cache placement happens in a much
larger timescale than scheduling and transmission \cite{Shanmugam2013,Binqiang2017,LiuD2018},
cache sizes optimization should be based on long-term
channel statistics.
Furthermore, to maximize the content downloading sum-rate,
the cache sizes should also be optimized together with the
optimal beamformers.
This gives the following cache size allocation problem:
\begin{subequations}\label{Cache0}
\begin{eqnarray}\label{Cache0 obj}
&\max_{\left\{C_k\right\}}&
\mathbb{E}_{\left\{\mathbf{H}_k\right\}}
\Bigg[\max_{\left\{\mathbf{V}_g\right\}}\sum_{g=1}^G\min_{k\in\mathcal{K}_g}
\Bigg\{\frac{F_g}{F_g-C_k}\nonumber\\
&&\cdot\log\det\left(\mathbf{I}_N
+\mathbf{H}_k\mathbf{V}_{g}\mathbf{V}_{g}^H\mathbf{H}_k^H
\mathbf{J}_k\right)\Bigg\}\Bigg],~~~~~
\end{eqnarray}
\begin{equation}\label{Cache0 power}
\text{s.t.}\
\sum_{g=1}^G\mathrm{Tr}\left(\mathbf{V}_{g}
\mathbf{V}_{g}^H\right)\leq P_{\text{tot}},
\end{equation}
\begin{equation}\label{Cache0 cache1}
~~~~\
\sum_{k=1}^KC_k\leq C_{\text{tot}},
\end{equation}
\begin{equation}\label{Cache0 cache2}
~~~~\
0\leq C_k \leq F_{g_k},~~\forall k=1, 2, \ldots, K,
\end{equation}
\end{subequations}
where $P_{\text{tot}}$ is the total budget of the transmit power at the computation center,
and $C_{\text{tot}}$ is the total budget of the cache size for the whole network.
A common approach to tackle the expectation in \eqref{Cache0 obj} is
the sample approximation \cite{Birge2011}, which
reformulates \eqref{Cache0} as
\begin{subequations}\label{Cache1}
\begin{eqnarray}\label{Cache1 obj}
&\max_{\left\{C_k\right\}, \left\{\mathbf{V}_{g,t}\right\}}&
\sum_{t=1}^T\sum_{g=1}^G\min_{k\in\mathcal{K}_g}
\Bigg\{\frac{F_g}{F_g-C_k}\log\det\Big(\mathbf{I}_N
\nonumber\\
&&+\mathbf{H}_{k,t}\mathbf{V}_{g,t}\mathbf{V}_{g,t}^H\mathbf{H}_{k,t}^H
\mathbf{J}_{k,t}\Big)\Bigg\},~~~~~~
\end{eqnarray}
\begin{equation}\label{Cache1 power}
\text{s.t.}\
\sum_{g=1}^G\mathrm{Tr}\left(\mathbf{V}_{g,t}
\mathbf{V}_{g,t}^H\right)\leq P_{\text{tot}},~~\forall t=1, 2, \ldots, T,
\end{equation}
\begin{equation}\label{Cache1 cache1}
~~~~\
\sum_{k=1}^KC_k\leq C_{\text{tot}},
\end{equation}
\begin{equation}\label{Cache1 cache2}
~~~~\
0\leq C_k \leq F_{g_k},~~\forall k=1, 2, \ldots, K,
\end{equation}
\end{subequations}
where $T$ is the sample size, $\left\{\mathbf{H}_{k,t}\right\}_{k=1}^K$
are the $t$-th channel samples,
which can be drawn from any given channel distribution or historical channel realizations,
$\left\{\mathbf{V}_{g,t}\right\}_{g=1}^G$ are the corresponding beamformers,
and $\mathbf{J}_{k,t}\triangleq\left(\sum_{g^\prime\neq g_k}\mathbf{H}_{k,t}\mathbf{V}_{g^\prime,t}\mathbf{V}_{g^\prime,t}^H\mathbf{H}_{k,t}^H
+\sigma_k^2\mathbf{I}_{N}\right)^{-1}$.

However, problem \eqref{Cache1} is challenging to solve due to three reasons.
Firstly, the objective function \eqref{Cache1 obj} is nonsmooth since the
content downloading rate of each BS cluster is the minimum over
$\left\vert\mathcal{K}_g\right\vert$ terms. Secondly, the objective function \eqref{Cache1 obj}
is also nonconcave due to the nonconcave coupling between $\frac{F_g}{F_g-C_k}$ and $\log\det\left(\mathbf{I}_N
+\mathbf{H}_{k,t}\mathbf{V}_{g,t}\mathbf{V}_{g,t}^H\mathbf{H}_{k,t}^H
\mathbf{J}_{k,t}\right)$,
and the involvement of $\left\{\mathbf{V}_{g,t}\right\}$ in the inter-cluster interference inside the expression of $\mathbf{J}_{k,t}$.
Thirdly, since the sample size $T$ is generally large for good approximation,
problem \eqref{Cache1} is imposed by a
large number of variables and constraints, which induces a heavy computational burden.

\vspace{2mm}
\textcolor{black}{
\emph{Remark 1:}
For multicast transmission, the file requests in the same multicast group
should arrive within a very short time period.
Although this assumption is a little strong for mobile users,
it makes more sense for the considered scenario in this paper,
where the multicast receivers are BSs
\textcolor{black}{that require data sharing for coordinated multipoint joint processing}
rather than mobile users.
Therefore, it is reasonable to assume that the BSs in the same cluster
request the content simultaneously.}

\vspace{2mm}
\textcolor{black}{
\emph{Remark 2:}
If the computation center transmits data to each BS
directly, it either transmits the data in a time-division fashion, or transmits multiple beams at the same time. For using the time division transmission, it avoids interference among beams, but it would take a long time to transmit, as one BS is served after another. On the other hand, if multiple beams are transmitted at the same time, the transmission time is shortened, but the interference among beams would cause severe decoding error.}

\vspace{2mm}
\textcolor{black}{
\emph{Remark 3:}
Strictly speaking, $C_k$ is a discrete variable, which makes
the optimization problem~\eqref{Cache1}
combinatorial and highly complex.
To make it more tractable, as in the most relevant work \cite{Dai2018}, we set $C_k$ as a continuous variable.
Consequently, it is much easier to reveal the insight of caching as shown in Section VI.
For practical implementation of the scheme, the solution of $C_k$ can be rounded off to the nearest integer after the continuous optimization problem is solved.}

\section{Accelerated First-Order Algorithm for Cache Allocation}
In this section, we strive to solve the large-scale nonsmooth nonconvex
cache size allocation problem \eqref{Cache1}.
Specifically, we first tackle the non-smoothness and non-convexity of problem
\eqref{Cache1} by introducing auxiliary variables and constructing a sequence of
quadratic convex functions in the SCA framework. Then,
instead of directly solving each convexified problem with the interior-point method,
we further construct a strongly convex upper bound of the cost function, so that an accelerated first-order algorithm is developed
for solving each SCA subproblem in its dual domain.

\subsection{Tackling Non-Smoothness and Non-Convexity}
We first tackle the non-smoothness of the objective function \eqref{Cache1 obj}.
Since \eqref{Cache1 obj} is the minimum over
$\left\vert\mathcal{K}_g\right\vert$ terms, by introducing a set of auxiliary variables
$\left\{\eta_{g,t}\right\}$ such that $\eta_{g,t}\leq
\frac{1}{F_g-C_k}\log\det\left(\mathbf{I}_N
+\mathbf{H}_{k,t}\mathbf{V}_{g,t}\mathbf{V}_{g,t}^H\mathbf{H}_{k,t}^H
\mathbf{J}_{k,t}\right)$, $\forall k\in\mathcal{K}_g$,
problem \eqref{Cache1} can be equivalently transformed into the following smooth problem:
\begin{subequations}\label{Cache2}
\begin{equation}\label{Cache2 obj}
\min_{\left\{C_{k}\right\}, \left\{\mathbf{V}_{g,t}, \eta_{g,t}\right\}}-\sum_{t=1}^T\sum_{g=1}^GF_g\eta_{g,t},
\end{equation}
\begin{equation}\label{Cache2 power}
\text{s.t.}\
\sum_{g=1}^G\mathrm{Tr}\left(\mathbf{V}_{g,t}
\mathbf{V}_{g,t}^H\right)\leq P_{\text{tot}},~~\forall t=1, 2, \ldots, T,
\end{equation}
\begin{equation}\label{Cache2 cache1}
~~~~\
\sum_{k=1}^KC_k\leq C_{\text{tot}},
\end{equation}
\begin{equation}\label{Cache2 cache2}
~~~~\
0\leq C_k \leq F_{g_k},~~\forall k=1, 2, \ldots, K,
\end{equation}
\begin{eqnarray}\label{Cache2 rate}
~~~~\
(F_{g_k}-C_k)\eta_{g_k,t}-\log\det\Big(\mathbf{I}_N
+\mathbf{H}_{k,t}\mathbf{V}_{g_k,t}
\mathbf{V}_{g_k,t}^H
\nonumber\\
\cdot\mathbf{H}_{k,t}^H
\mathbf{J}_{k,t}\Big)
\leq0, \forall k=1, 2, \ldots, K,
\forall t=1, 2, \ldots, T.
\end{eqnarray}
\end{subequations}

However, due to the nonconvex coupling between $C_k$ and $\eta_{g_k,t}$,
and the involvement of $\left\{\mathbf{V}_{g,t}\right\}$ in the inter-cluster interference inside the expression of $\mathbf{J}_{k,t}$, the constraint \eqref{Cache2 rate} is nonconvex, making
problem \eqref{Cache2} still challenging to solve.
For the special case of $G=1$, $\mathbf{J}_k$ would reduce to $1/\sigma_k^2\mathbf{I}_N$,
which is independent of the variable $\mathbf{V}_1$. Therefore,
by introducing an auxiliary variable $\mathbf{W}_1=\mathbf{V}_1\mathbf{V}_1^H$
and dropping the rank constraint of $\mathbf{W}_1$ \cite{Dai2018},
\eqref{Cache2 rate} would be convex over $\mathbf{W}_1$.
However, for the general case of $G>1$, since $\mathbf{J}_k$
involves variables $\left\{\mathbf{V}_g\right\}$,
the above convexity over $\left\{\mathbf{W}_g\right\}$ does not hold.
Thus, for the general setting of $G>1$,
the non-convexity of \eqref{Cache2 rate} is difficult to tackle.

A prevalent technique to tackle nonconvex constraints is the successive convex approximation (SCA) \cite{Beck2010}, in which nonconvex constraints are approximated by a sequence of convex constraints. When the nonconvex constraints are in difference of convex (DC) forms, a common approach for convex approximation is the convex-concave procedure (CCP) \cite{Lipp2016}.
Nevertheless, CCP is not applicable to \eqref{Cache2 rate}, since it is not in a DC form.
To address this issue, we construct a sequence of convex constraints to approximate \eqref{Cache2 rate} by quadratically convexifying the left-hand-side of \eqref{Cache2 rate}.
Specifically, given any fixed $C_{k}^{(i)}$, $\eta_{g_k,t}^{(i)}$, and
$\left\{\mathbf{V}_{g,t}^{(i)}\right\}_{g=1}^G$, we define a convex quadratic function
\begin{eqnarray}\label{f2}
f_{k,t}^{(i)}\left(C_{k}, \eta_{g_k,t}, \left\{\mathbf{V}_{g,t}\right\}_{g=1}^G\right)
\triangleq
\sum_{g=1}^G\mathrm{Tr}\left(\mathbf{V}_{g,t}^H
\mathbf{A}_{k,t}^{(i)}\mathbf{V}_{g,t}\right)
\nonumber\\
+2\Re\left\{\mathrm{Tr}\left(\mathbf{B}_{k,t}^{(i)}\mathbf{V}_{g_k,t}\right)\right\}
+\frac{\eta_{g_k,t}^2+C_k^2}{2}
+F_{g_k}\eta_{g_k,t}
\nonumber\\
-\left(\eta_{g_k,t}^{(i)}+C_k^{(i)}\right)\left(\eta_{g_k,t}+C_k\right)
+b_{k,t}^{(i)},
\end{eqnarray}
where $\mathbf{A}_{k,t}^{(i)}$, $\mathbf{B}_{k,t}^{(i)}$, and $b_{k,t}^{(i)}$
are given by
\begin{eqnarray}
\mathbf{A}_{k,t}^{(i)}&\triangleq&
\mathbf{H}_{k,t}^H\mathbf{U}_{k,t}^{(i)}\left(\mathbf{I}_{d}-\left(\mathbf{U}_{k,t}^{(i)}\right)^H
\mathbf{H}_{k,t}\mathbf{V}_{g_k,t}^{(i)}\right)^{-1}\nonumber\\
&&\cdot\left(\mathbf{U}_{k,t}^{(i)}\right)^H\mathbf{H}_{k,t},
\label{Dt}\\
\mathbf{B}_{k,t}^{(i)}&\triangleq&
-\left(\mathbf{I}_{d}-\left(\mathbf{U}_{k,t}^{(i)}\right)^H
\mathbf{H}_{k,t}\mathbf{V}_{g_k,t}^{(i)}\right)^{-1}
\nonumber\\
&&\cdot\left(\mathbf{U}_{k,t}^{(i)}\right)^H\mathbf{H}_{k,t},
\label{Ct}\\
b_{k,t}^{(i)}&\triangleq&
\mathrm{Tr}\Bigg(\left(\mathbf{I}_{d}-\left(\mathbf{U}_{k,t}^{(i)}\right)^H
\mathbf{H}_{k,t}\mathbf{V}_{g_k,t}^{(i)}\right)^{-1}
\nonumber\\
&&\cdot
\left(\mathbf{I}_{d}
+\sigma_k^2\left(\mathbf{U}_{k,t}^{(i)}\right)^H\mathbf{U}_{k,t}^{(i)}\right)\Bigg)
\nonumber\\
&&+\log\det\left(\mathbf{I}_d-\left(\mathbf{U}_{k,t}^{(i)}\right)^H
\mathbf{H}_{k,t}\mathbf{V}_{g_k,t}^{(i)}\right)
\nonumber\\
&&
+\frac{\left(\eta_{g_k,t}^{(i)}+C_k^{(i)}\right)^2}{2}-d,
\label{bt}
\end{eqnarray}
with
\begin{eqnarray}\label{Uti}
\mathbf{U}_{k,t}^{(i)}&\triangleq&
\left(\sum_{g=1}^G\mathbf{H}_{k,t}\mathbf{V}_{g,t}^{(i)}
\left(\mathbf{V}_{g,t}^{(i)}\right)^H\mathbf{H}_{k,t}^H
+\sigma_k^2\mathbf{I}_{N}\right)^{-1}
\nonumber\\
&&\cdot\mathbf{H}_{k,t}\mathbf{V}_{g_k,t}^{(i)}.
\end{eqnarray}
Then, we can establish
two properties of $f_{k,t}^{(i)}\left(C_{k}, \eta_{g_k,t}, \left\{\mathbf{V}_{g,t}\right\}_{g=1}^G\right)$ with the following proposition.

\newtheorem{proposition}{Proposition}
\begin{proposition}\label{lm2}
The defined function $f_{k,t}^{(i)}\left(C_{k}, \eta_{g_k,t}, \left\{\mathbf{V}_{g,t}\right\}_{g=1}^G\right)$ in \eqref{f2} satisfies:

\textbf{(1.1)} $f_{k,t}^{(i)}\left(C_{k}, \eta_{g_k,t}, \left\{\mathbf{V}_{g,t}\right\}_{g=1}^G\right)\geq(F_{g_k}-C_k)\eta_{g_k,t}-\log\det\left(\mathbf{I}_N
+\mathbf{H}_{k,t}\mathbf{V}_{g_k,t}\mathbf{V}_{g_k,t}^H\mathbf{H}_{k,t}^H
\mathbf{J}_{k,t}\right)$,
where the equality holds at $C_k=C_k^{(i)}$, $\eta_{g_k}=\eta_{g_k}^{(i)}$, and
$\mathbf{V}_{g,t}=\mathbf{V}_{g,t}^{(i)}$, $\forall g=1,2,\ldots,G$.

\textbf{(1.2)} $\frac{\partial }{\partial \tau}
f_{k,t}^{(i)}\left(C_{k}^{(i)}, \eta_{g_k,t}^{(i)}, \left\{\mathbf{V}_{g,t}^{(i)}\right\}_{g=1}^G\right)=\frac{\partial}{\partial \tau}\Bigg((F_{g_k}-C_k^{(i)})\eta_{g_k,t}^{(i)}-\log\det\Bigg(\mathbf{I}_N
+\mathbf{H}_{k,t}\mathbf{V}_{g_k,t}^{(i)}\left(\mathbf{V}_{g_k,t}^{(i)}\right)^H$
$\mathbf{H}_{k,t}^H
\mathbf{J}_{k,t}^{(i)}\Bigg)\Bigg)$,
where $\tau$ represents $C_k$, $\eta_{g_k}$, or
any element of $\mathbf{V}_{g,t}$, $\forall g=1,2,\ldots,G$,
and $\mathbf{J}_{k,t}^{(i)}=\left.\mathbf{J}_{k,t}\right\vert_{\big\{\mathbf{V}_{g,t}\big\}_{g=1}^G=\big\{\mathbf{V}_{g,t}^{(i)}\big\}_{g=1}^G}$.
\end{proposition}
\begin{proof}
See Appendix \ref{proof of lm2}.
\end{proof}

\vspace{2mm}

\textcolor{black}
{In particular, property $\textbf{(1.1)}$
means that the left-hand-sides of the original nonconvex constraints
are upper bounded by the left-hand-sides of
the constructed convex constraints;
while property $\textbf{(1.2)}$ means that
the gradients of the left-hand-sides of both the constructed convex constraints
and the original nonconvex constraints are equal
at the expansion points.}
Based on the two properties in Proposition \ref{lm2},
the left-hand-side of
the nonconvex constraint \eqref{Cache2 rate} is upper bounded by
$f_{k,t}^{(i)}\left(C_{k}, \eta_{g_k,t}, \left\{\mathbf{V}_{g,t}\right\}_{g=1}^G\right)$,
and hence \eqref{Cache2 rate}
can be successively approximated by
\begin{eqnarray}\label{convex c2}
f_{k,t}^{(i)}\left(C_{k}, \eta_{g_k,t}, \left\{\mathbf{V}_{g,t}\right\}_{g=1}^G\right)\leq 0,
\nonumber\\
~~\forall k=1, 2,
\ldots, K,~~\forall t=1, 2, \ldots, T,
\end{eqnarray}
which is convex since $f_{k,t}^{(i)}\left(C_{k}, \eta_{g_k,t}, \left\{\mathbf{V}_{g,t}\right\}_{g=1}^G\right)$ is convex quadratic.
With the sequence of convex constraints constructed in \eqref{convex c2},
problem \eqref{Cache2} can be iteratively solved in the SCA framework,
with the $i$-th SCA subproblem written as
\begin{subequations}\label{Cache2-2}
\begin{eqnarray}\label{Cache2-2 obj}
&&\left[\left\{C_k^{(i+1)}\right\},\left\{\mathbf{V}_{g,t}^{(i+1)}, \eta_{g,t}^{(i+1)}\right\}\right]
\nonumber\\
&=&
\arg\min_{\left\{C_{k}\right\}, \left\{\mathbf{V}_{g,t}, \eta_{g,t}\right\}}-\sum_{t=1}^T\sum_{g=1}^GF_g\eta_{g,t},
\end{eqnarray}
\begin{equation}\label{Cache3 power}
\text{s.t.}\
\sum_{g=1}^G\mathrm{Tr}\left(\mathbf{V}_{g,t}
\mathbf{V}_{g,t}^H\right)\leq P_{\text{tot}},~~\forall t=1, 2, \ldots, T,
\end{equation}
\begin{equation}\label{Cache3 cache1}
~~~~\
\sum_{k=1}^KC_k\leq C_{\text{tot}},
\end{equation}
\begin{equation}\label{Cache3 cache2}
~~~~\
0\leq C_k \leq F_{g_k},~~\forall k=1, 2, \ldots, K,
\end{equation}
\begin{eqnarray}\label{Cache3 rate}
~~~~\
f_{k,t}^{(i)}\left(C_{k}, \eta_{g_k,t}, \left\{\mathbf{V}_{g,t}\right\}_{g=1}^G\right)\leq 0,
\nonumber\\
~~\forall k=1, 2, \ldots, K,~~\forall t=1, 2, \ldots, T.
\end{eqnarray}
\end{subequations}

\subsection{First-Order Algorithm in Dual Domain}
While problem \eqref{Cache2-2} can be optimally solved with the interior point method, due to the
large number of variables and constraints (induced by the large sample size $T$),
such a method would incur a heavy computational cost.
To avoid such a heavy computational burden, we strive to develop a first-order algorithm,
which alternatively performs a gradient step and a projection step.
However, since problem \eqref{Cache2-2} is imposed by coupling constraints \eqref{Cache3 power}, \eqref{Cache3 cache1}, and \eqref{Cache3 rate}, the projection onto
\eqref{Cache3 power}-\eqref{Cache3 rate} would be highly complicated.

To address this issue,
we develop another form of the $i$-th SCA subproblem of \eqref{Cache2} by majorizing
the cost function \eqref{Cache2 obj}
with a strongly convex upper bound. Specifically, given any fixed
$\left\{C_{k}^{(i)}\right\}$ and $\left\{\mathbf{V}_{g,t}^{(i)}, \eta_{g,t}^{(i)}\right\}$,
the cost function \eqref{Cache2 obj} can be strongly convexified
by adding three positive quadratic terms:
\begin{eqnarray}\label{strong convex}
&&\Upsilon^{(i)}\left(\left\{C_{k}\right\}, \left\{\mathbf{V}_{g,t}, \eta_{g,t}\right\}\right)
\nonumber\\
&=&
-\sum_{t=1}^T\sum_{g=1}^GF_g\eta_{g,t}
+\frac{\rho_1}{2}\sum_{t=1}^T\sum_{g=1}^G\left(\eta_{g,t}-\eta_{g,t}^{(i)}\right)^2
\nonumber\\&&+\rho_2\sum_{t=1}^T\sum_{g=1}^G\left\Vert\mathbf{V}_{g,t}-\mathbf{V}_{g,t}^{(i)}\right\Vert_F^2
\nonumber\\&&
+\frac{\rho_3}{2}\sum_{k=1}^K\left(C_k-C_k^{(i)}\right)^2,
\end{eqnarray}
where $\rho_1$, $\rho_2$, and $\rho_3$ are fixed positive parameters. Consequently,
\eqref{strong convex} serves as a tight upper bound of \eqref{Cache2 obj},
with their function values equal at $\left\{C_{k}\right\}=\left\{C_{k}^{(i)}\right\}$ and $\left\{\mathbf{V}_{g,t}, \eta_{g,t}\right\}=\left\{\mathbf{V}_{g,t}^{(i)}, \eta_{g,t}^{(i)}\right\}$.
Following the same procedure for convexifying the constraints of \eqref{Cache2} in the last section,
another valid $i$-th SCA subproblem of \eqref{Cache2} can be written as
\begin{eqnarray}\label{Cache3}
&&\left[\left\{C_k^{(i+1)}\right\},\left\{\mathbf{V}_{g,t}^{(i+1)}, \eta_{g,t}^{(i+1)}\right\}\right]
\nonumber\\
&=&\arg\min_{\left\{C_{k}\right\}, \left\{\mathbf{V}_{g,t}, \eta_{g,t}\right\}}\Upsilon^{(i)}\left(\left\{C_{k}\right\}, \left\{\mathbf{V}_{g,t}, \eta_{g,t}\right\}\right),\\
&&\text{s.t.}~~\eqref{Cache3 power}, \eqref{Cache3 cache1},
\eqref{Cache3 cache2}, \eqref{Cache3 rate}.\nonumber
\end{eqnarray}
Based on the strong convexity of $\Upsilon^{(i)}\left(\left\{C_{k}\right\}, \left\{\mathbf{V}_{g,t}, \eta_{g,t}\right\}\right)$, we can derive the dual problem of \eqref{Cache3} in closed-form with the following proposition.

\begin{proposition}\label{smooth}
The dual problem of \eqref{Cache3} is
\begin{subequations}\label{dual p}
\begin{eqnarray}\label{dual}
\max_{\left\{\delta_{t}\right\}, \left\{\lambda_{k,t}\right\}, \mu}
&&\Upsilon^{(i)}\left(\left\{C_{k}^{\diamond}\right\}, \left\{\mathbf{V}_{g,t}^{\diamond}, \eta_{g,t}^{\diamond}\right\}\right)
\nonumber\\
&&+\sum_{t=1}^T\delta_t\left(\sum_{g=1}^G\mathrm{Tr}\left(\mathbf{V}_{g,t}^{\diamond}
\left(\mathbf{V}_{g,t}^{\diamond}\right)^H\right)-P_{\text{tot}}\right)
\nonumber\\&&+\sum_{t=1}^T\sum_{k=1}^K
\lambda_{k,t}f_{k,t}^{(i)}\left(C_{k}^{\diamond}, \eta_{g_k,t}^{\diamond}, \left\{\mathbf{V}_{g,t}^{\diamond}\right\}_{g=1}^G\right)
\nonumber\\
&&+\mu\left(\sum_{k=1}^KC_k^{\diamond}-C_{\text{tot}}\right),
\end{eqnarray}
\begin{eqnarray}\label{dual2}
\text{s.t.}\
\mu\geq0,~~
\delta_{t}\geq0,~~\lambda_{k,t}\geq0,
\nonumber\\
~~\forall k=1, 2, \ldots, K,~~\forall t=1, 2, \ldots, T,
\end{eqnarray}
\end{subequations}
where
$\left\{\mathbf{V}_{g,t}^{\diamond}, \eta_{g,t}^{\diamond}\right\}$
and $\left\{C_{k}^{\diamond}\right\}$ are uniquely given in the following closed forms:
\begin{eqnarray}\label{eta}
\eta_{g,t}^{\diamond}=\eta_{g,t}^{(i)}+\frac{\left(1-\sum_{k\in\mathcal{K}_g}\lambda_{k,t}\right)F_g+
\sum_{k\in\mathcal{K}_g}\lambda_{k,t}C_k^{(i)}}{\rho_1+\sum_{k\in\mathcal{K}_g}\lambda_{k,t}},
\nonumber\\
~~\forall g=1,2,\ldots,G,
~~\forall t=1,2,\ldots,T,
\end{eqnarray}
\begin{eqnarray}\label{V}
\mathbf{V}_{g,t}^{\diamond}&=&\left(\left(\rho_2+\delta_t\right)\mathbf{I}_M
+\sum_{k=1}^K\lambda_{k,t}\mathbf{A}_{k,t}^{(i)}\right)^{-1}
\nonumber\\
&&\cdot\left(\rho_2\mathbf{V}_{g,t}^{(i)}-\sum_{k\in\mathcal{K}_g}\lambda_{k,t}\left(\mathbf{B}_{k,t}^{(i)}
\right)^H\right),\nonumber\\
&&\forall g=1,2,\ldots,G,
~~\forall t=1,2,\ldots,T,
\end{eqnarray}
\begin{eqnarray}\label{Ck}
&&C_k^{\diamond}=\min\Bigg\{\max\Bigg\{\Bigg(C_k^{(i)}+
\nonumber\\
&&\frac{\sum_{t=1}^T\lambda_{k,t}\eta_{g_k,t}^{(i)}-\mu}
{\rho_3+\sum_{t=1}^T\lambda_{k,t}}\Bigg), 0\Bigg\},
F_{g_k}\Bigg\},
\forall k=1,2,\ldots,K.~~~~~
\end{eqnarray}
\end{proposition}

\begin{proof}
See Appendix \ref{proof of smooth}.
\end{proof}

\vspace{2mm}

Denoting the the dual objective
function in \eqref{dual} as $D\left(\left\{\delta_{t}\right\}, \left\{\lambda_{k,t}\right\}, \mu\right)$,
and noticing that the values of $\left\{\mathbf{V}_{g,t}^{\diamond}, \eta_{g,t}^{\diamond}\right\}$
and $\left\{C_{k}^{\diamond}\right\}$
are uniquely determined by \eqref{eta}-\eqref{Ck},
we can obtain the partial derivatives as
\begin{equation}\label{gradient}
\begin{cases}
\frac{\partial D}{\partial \delta_t}=\sum_{g=1}^G\mathrm{Tr}\left(\mathbf{V}^{\diamond}_{g,t}
\left(\mathbf{V}^{\diamond}_{g,t}\right)^H\right)-P_{\text{tot}},
\cr \frac{\partial D}{\partial \lambda_{k,t}}=f_{k,t}^{(i)}\left(C_{k}^{\diamond}, \eta_{g_k,t}^{\diamond}, \left\{\mathbf{V}_{g,t}^{\diamond}\right\}_{g=1}^G\right),
\cr \frac{\partial D}{\partial \mu}=\sum_{k=1}^KC_k^{\diamond}-C_{\text{tot}},
\end{cases}
\end{equation}
thus the projected gradient step increasing $D\left(\left\{\delta_{t}\right\}, \left\{\lambda_{k,t}\right\}, \mu\right)$
can be expressed as
\begin{eqnarray}\label{gstep}
\begin{cases}
\left(\delta_{t}+\beta_s \left(\sum_{g=1}^G\mathrm{Tr}\left(\mathbf{V}^{\diamond}_{g,t}
\left(\mathbf{V}^{\diamond}_{g,t}\right)^H\right)
-P_{\text{tot}}\right)\right)^{+},
\cr
\left(\lambda_{k,t}+\beta_s f_{k,t}^{(i)}\left(C_{k}^{\diamond}, \eta_{g_k,t}^{\diamond}, \left\{\mathbf{V}_{g,t}^{\diamond}\right\}_{g=1}^G\right)\right)^{+},
\cr
\left(\mu+\beta_s \left(\sum_{k=1}^KC_k^{\diamond}-C_{\text{tot}}\right)\right)^{+},
\end{cases}
\end{eqnarray}
where $\beta_s$ is the step size at the $s$-th iteration, and $(\cdot)^+\triangleq\max(\cdot, 0)$ is the non-negativity projection over \eqref{dual2}.

\begin{algorithm}[t!]
\caption{First-Order Algorithm for Solving \eqref{Cache3}}
\begin{algorithmic}[1]\footnotesize
\State Compute $\mathbf{A}_{k,t}^{(i)}$, $\mathbf{B}_{k,t}^{(i)}$, and $b_{k,t}^{(i)}$ with \eqref{Dt}-\eqref{bt},
$\forall k=1,2,\ldots,K,~~\forall t=1,2,\ldots,T.$\\
Initialize $\delta_t=1$, $\forall t=1,2,\ldots,T$;
$\lambda_{k,t}=1$, $\forall k=1,2,\ldots,K$, $\forall t=1,2,\ldots,T$; $\mu=1$.\\
$\textbf{repeat}$ ($s=1,2,\ldots$)\\
Update $\left\{\mathbf{V}^{\diamond}_{g,t}, \eta^{\diamond}_{g,t}\right\}$
and $\left\{C^{\diamond}_{k}\right\}$
with \eqref{eta}-\eqref{Ck}.\\
Update $\left\{\delta_{t}\right\}$, $\left\{\lambda_{k,t}\right\}$, and $\mu$
with \eqref{gstep}.\\
$\textbf{until}$ convergence\\
Output $\left\{\mathbf{V}^{(i+1)}_{g,t}, \eta^{(i+1)}_{g,t}\right\}=\left\{\mathbf{V}^{\diamond}_{g,t}, \eta^{\diamond}_{g,t}\right\}$
and $\left\{C^{(i+1)}_{k}\right\}=\left\{C^{\diamond}_{k}\right\}$.
\end{algorithmic}
\end{algorithm}

By iteratively updating  $\left\{\delta_{t}\right\}$, $\left\{\lambda_{k,t}\right\}$, and $\mu$ with
\eqref{gstep}, we can obtain the optimal solution to the dual problem \eqref{dual p}.
Correspondingly, the optimal $\left\{\mathbf{V}_{g,t}, \eta_{g,t}\right\}$
and $\left\{C_{k}\right\}$ to the primal problem \eqref{Cache3} is given by substituting the optimal
$\left\{\delta_{t}\right\}$, $\left\{\lambda_{k,t}\right\}$, and $\mu$ into
\eqref{eta}-\eqref{Ck}. We summarize this procedure for solving problem \eqref{Cache3}
in Algorithm 1, which is guaranteed to
converge to the global optimum of \eqref{dual p} at a rate of $\mathcal{O}\left(1/s\right)$, if the step size $\beta_s$ is smaller than the inverse of the Lipschitz constant of $\nabla D$ \cite{Beck2009-2}. Moreover, notice that the primal problem \eqref{Cache3} is convex, thus
the convergent optimum of \eqref{dual p} is also the global optimum of \eqref{Cache3}, provided that
\eqref{Cache3} is strictly feasible~\cite{Boyd2004}.

\begin{figure}[t!]
\begin{center}
  \includegraphics[width=0.35\textwidth]{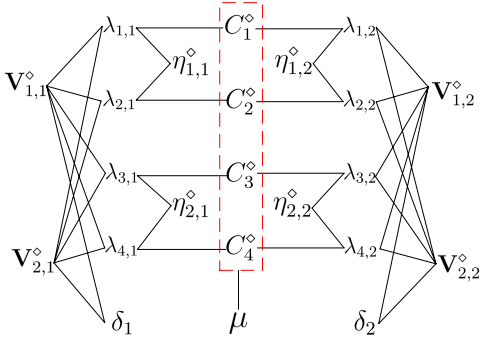}
  \caption{Parallel structure of Algorithm 1 when $T=2$, $G=2$, $K=4$, $\mathcal{K}_1=\left\{1,2\right\}$, and $\mathcal{K}_2=\left\{3,4\right\}$.}\label{PARALLEL}
\end{center}
\vspace{-0.2cm}
\end{figure}

Notice that each iteration of Algorithm 1 can be executed in parallel. In particular,
both the $2GT+K$ primal variables in line 4
and the $(K+1)T+1$ dual variables in line 5 can be updated in parallel.
An example of the parallel structure of Algorithm 1
is shown in Fig. \ref{PARALLEL}, where
$T=2$, $G=2$, $K=4$, $\mathcal{K}_1=\left\{1,2\right\}$, and $\mathcal{K}_2=\left\{3,4\right\}$.
The $12$ primal variables $\mathbf{V}^{\diamond}_{1,1}$,
$\mathbf{V}^{\diamond}_{2,1}$, $\mathbf{V}^{\diamond}_{1,2}$, $\mathbf{V}^{\diamond}_{2,2}$, $\eta_{1,1}^{\diamond}$,
$\eta_{2,1}^{\diamond}$,  $\eta_{1,2}^{\diamond}$,  $\eta_{2,2}^{\diamond}$,
$C_1^{\diamond}$, $C_2^{\diamond}$, $C_3^{\diamond}$, and $C_4^{\diamond}$ can be simultaneously updated. Furthermore, the update of each primal variable only depends on a few dual variables (e.g.,
the update of $\eta_{1,1}^{\diamond}$ only
depends on $\lambda_{1,1}$ and $\lambda_{2,1}$), thus the message passing overhead is small.
Similarly, the dual variables can also be simultaneously updated and each
depends on only a few primal variables.
Due to this parallel structure, Algorithm 1 has the potential of leveraging the modern multi-core multi-thread processor architecture for speeding up the computation.

\subsection{Acceleration with Momentum Technique}
\begin{algorithm}[t!]
\caption{Accelerated First-Order Algorithm for Solving \eqref{Cache3}}
\begin{algorithmic}[1]\footnotesize
\State Compute $\mathbf{A}_{k,t}^{(i)}$, $\mathbf{B}_{k,t}^{(i)}$, and $b_{k,t}^{(i)}$ with \eqref{Dt}-\eqref{bt},
$\forall k=1,2,\ldots,K,~~\forall t=1,2,\ldots,T.$\\
Initialize $\delta_t=1$, $\forall t=1,2,\ldots,T$;
$\lambda_{k,t}=1$, $\forall k=1,2,\ldots,K$, $\forall t=1,2,\ldots,T$; $\mu=1$; $\theta^{(0)}=1.$\\
$\textbf{repeat}$ ($s=1,2,\ldots$)\\
Update $\theta^{(s)}$ with \eqref{theta}.\\
Update $\left\{\mathbf{V}^{\diamond}_{g,t}, \eta^{\diamond}_{g,t}\right\}$
and $\left\{C^{\diamond}_{k}\right\}$
with \eqref{eta}-\eqref{Ck}.\\
Update $\left\{\tilde{\delta}_{t}^{(s)}\right\}$, $\left\{\tilde{\lambda}_{k,t}^{(s)}\right\}$, and $\tilde{\mu}^{(s)}$
with \eqref{gstep1}.\\
Update $\left\{\delta_{t}\right\}$, $\left\{\lambda_{k,t}\right\}$, and $\mu$
with \eqref{gstep2}.\\
$\textbf{until}$ convergence\\
Output $\left\{\mathbf{V}^{(i+1)}_{g,t}, \eta^{(i+1)}_{g,t}\right\}=\left\{\mathbf{V}^{\diamond}_{g,t}, \eta^{\diamond}_{g,t}\right\}$
and $\left\{C^{(i+1)}_{k}\right\}=\left\{C^{\diamond}_{k}\right\}$.
\end{algorithmic}
\end{algorithm}

Although Algorithm 1 only involves the gradient information and can be executed in parallel,
as a first-order algorithm, it may require a large number of iterations to converge.
To improve the convergence speed, we further apply the
momentum technique \cite{Beck2009-3} to accelerate Algorithm 1.
In particular, the projected gradient step in \eqref{gstep}
is modified by updating $\tilde{\delta}_t^{(s)}$, $\tilde{\lambda}_{k,t}^{(s)}$,
and $\tilde{\mu}^{(s)}$:
\begin{eqnarray}\label{gstep1}
\begin{cases}
\left(\delta_{t}+\beta_s \left(\sum_{g=1}^G\mathrm{Tr}\left(\mathbf{V}^{\diamond}_{g,t}
\left(\mathbf{V}^{\diamond}_{g,t}\right)^H\right)-P_{\text{tot}}\right)\right)^{+},
\cr
\left(\lambda_{k,t}+\beta_s f_{k,t}^{(i)}\left(C_{k}^{\diamond}, \eta_{g_k,t}^{\diamond}, \left\{\mathbf{V}_{g,t}^{\diamond}\right\}_{g=1}^G\right)\right)^{+},
\cr
\left(\mu+\beta_s \left(\sum_{k=1}^KC_k^{\diamond}-C_{\text{tot}}\right)\right)^{+},
\end{cases}
\end{eqnarray}
\begin{eqnarray}\label{gstep2}
\begin{cases}
\delta_t\leftarrow\tilde{\delta}_t^{(s)}+\frac{\theta^{(s-1)}-1}{\theta^{(s)}}\left(\tilde{\delta}_t^{(s)}-\tilde{\delta}_t^{(s-1)}\right),
\cr
\lambda_{k,t}\leftarrow\tilde{\lambda}_{k,t}^{(s)}+\frac{\theta^{(s-1)}-1}{\theta^{(s)}}
\left(\tilde{\lambda}_{k,t}^{(s)}-\tilde{\lambda}_{k,t}^{(s-1)}\right),
\cr
\mu\leftarrow\tilde{\mu}^{(s)}+\frac{\theta^{(s-1)}-1}{\theta^{(s)}}
\left(\tilde{\mu}^{(s)}-\tilde{\mu}^{(s-1)}\right),
\end{cases}
\end{eqnarray}
where $\theta^{(s)}$ is the weighting parameter to dynamically control the momentums $\tilde{\delta}_t^{(s)}-\tilde{\delta}_t^{(s-1)}$,
$\tilde{\lambda}_{k,t}^{(s)}-\tilde{\lambda}_{k,t}^{(s-1)}$, and
$\tilde{\mu}^{(s)}-\tilde{\mu}^{(s-1)}$. To achieve fast convergence,
$\theta^{(s)}$ is updated by \cite{Beck2009-3}
\begin{equation}\label{theta}
\theta^{(s)}=\frac{1+\sqrt{1+4\left(\theta^{(s-1)}\right)^2}}{2}.
\end{equation}
By using \eqref{gstep1}-\eqref{theta},
the accelerated first-order algorithm for solving problem \eqref{Cache3}
is summarized in Algorithm 2. The key insight of the acceleration lies in
the momentums $\tilde{\delta}_t^{(s)}-\tilde{\delta}_t^{(s-1)}$,
$\tilde{\lambda}_{k,t}^{(s)}-\tilde{\lambda}_{k,t}^{(s-1)}$, and
$\tilde{\mu}^{(s)}-\tilde{\mu}^{(s-1)}$ in \eqref{gstep2} (without these momentums,
Algorithm 2 would reduce to Algorithm 1). These momentums
utilize previous updates to generate an overshoot, so that the update using
\eqref{gstep1} and \eqref{gstep2} in Algorithm 2
is more aggressive than the conventional gradient step \eqref{gstep} in Algorithm 1.
On the other hand, to ensure these overshoots to be well behaved,
the momentums are controlled by a sequence of weighting parameters
$\left\{\theta^{(s)}\right\}$.
With $\left\{\theta^{(s)}\right\}$ updated using \eqref{theta}, Algorithm 2
is guaranteed to converge to the the global optimum
of \eqref{dual p} at a rate of $\mathcal{O}\left(1/s^2\right)$ \cite{Beck2009-3}.

\subsection{Overall Algorithm for Solving \eqref{Cache1}}
With the $i$-th SCA subproblem \eqref{Cache3} solved by Algorithm 2,
the overall algorithm for solving the cache size allocation problem \eqref{Cache1} is summarized in Algorithm 3.
Since the constructed convex constraint \eqref{convex c2} satisfies properties $\textbf{(1.1)}$ and $\textbf{(1.2)}$ of Proposition \ref{lm2},
and the constructed upper bound $\Upsilon^{(i)}\left(\left\{C_{k}\right\}, \left\{\mathbf{V}_{g,t}, \eta_{g,t}\right\}\right)$ tightly approximates the cost function \eqref{Cache2 obj},
Algorithm~3 is guaranteed to converge to
a stationary point of problem~\eqref{Cache2} \cite{Razaviyayn2014},
which is equivalent to problem~\eqref{Cache1}.

Notice that Algorithm 3 is based on the SCA framework, thus it requires to be
initialized from a feasible point.
For simplicity, as shown in line 1,
we provide a feasible initial point for Algorithm 3 by
equally allocating the total transmit power $P_{\text{tot}}$ and the total cache sizes $C_{\text{tot}}$ to each $\mathbf{V}_{g,t}^{(0)}$ and $C_k^{(0)}$ respectively, and
$\eta_{g,t}^{(0)}$ is obtained from \eqref{Cache2 rate} by substituting
$\mathbf{V}_{g,t}=\mathbf{V}_{g,t}^{(0)}$ and $C_k=C_k^{(0)}$.

\begin{algorithm}[t!]
\caption{Overall Algorithm for Solving \eqref{Cache1}}
\begin{algorithmic}[1]\footnotesize
\State Initialization:
\begin{enumerate}
\item[]
$\mathbf{V}_{g,t}^{(0)}=\sqrt{\frac{P_{\text{tot}}}{GMd}}\mathbf{1}_{M\times d},~~\forall g=1,2,\ldots,G,~~\forall t=1,2,\ldots,T$;
\item[]
$C_k^{(0)}=C_{\text{tot}}/K,~~\forall k=1,2,\ldots, K$;
\item[]
$\eta_{g,t}^{(0)}=\min_{k\in\mathcal{K}_g}
\Big\{\frac{1}{F_g-C_k^{(0)}}\log\det\Big(\mathbf{I}_N+\mathbf{H}_{k,t}\mathbf{V}_{g,t}^{(0)}$
$\cdot\left(\mathbf{V}_{g,t}^{(0)}\right)^H\mathbf{H}_{k,t}^H
\mathbf{J}_{k,t}^{(0)}\Big)\Big\},~~\forall g=1,2,\ldots,G,~~\forall t=1,2,\ldots,T$.
\end{enumerate}\\
$\textbf{repeat}$ ($i=0,1,\ldots$)\\
Solve the $i$-th SCA subproblem \eqref{Cache3} with Algorithm 2.\\
$\textbf{until}$ convergence
\end{algorithmic}
\end{algorithm}

\iffalse
\textcolor{black}{
While both \cite{Li2018} and this paper develop first-order algorithms,
they propose different techniques to exploit different structures of their own problems.
In particular, the algorithm framework in \cite{Li2018} is alternating minimization,
which decomposes the original problem into two subproblems.
By exploiting the strong duality and
block separable structure, each subproblem
is solved either in a closed form or via parallel decomposition.
In contrast, in this paper, the problem \eqref{Cache1} does not enjoy the above structures,
and hence cannot be solved by the proposed algorithm in \cite{Li2018}.
Instead, as shown in Fig.~\ref{f1-1},
we first tackle the non-smoothness and non-convexity of \eqref{Cache1} using Proposition~1.
However, even after this step, it is still challenging to develop a first-order algorithm
due to the large number of coupling constraints. To tackle this challenge,
we further establish Proposition~2 to exploit the strong convexity and then
solve the dual problem by the proposed accelerated gradient based method.
Compared with \cite{Li2018}, the technical contributions of this paper lie in the new technique,
which solves a different structured problem \eqref{Cache1} in a new algorithm framework as shown in Fig.~\ref{f1-1}.}
\fi

\begin{figure}[t!]
	\begin{center}
        {\includegraphics[width=0.47\textwidth]{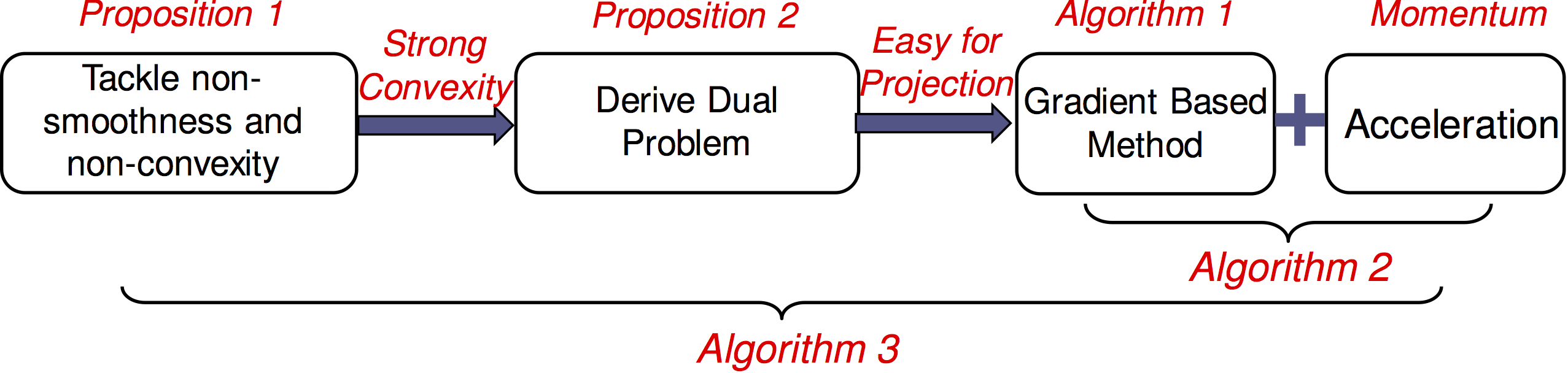}
		\caption{Proposed algorithm framework for the cache allocation problem \eqref{Cache1}.}\label{f1-1}}
	\end{center}
	\vspace{-0.2cm}
\end{figure}

\textcolor{black}{
In summary,
the extension from the single-cluster scenario [8]
to the more general multi-cluster scenario
brings two new challenges. First, the inter-cluster interference induces non-convexity in the optimization problem. Although SCA provides a general framework to tackle the non-convexity, it is still challenging to convexify the nonconvex constraint \eqref{Cache2 rate}, which is not in the common difference of convex form. To tackle this challenge, we establish Proposition~1,
which tightly approximates \eqref{Cache2 rate} by quadratically convexifying its left-hand-side.
Second, even after tackling the non-convexity, it is still challenging to solve the resulting SCA subproblem
\eqref{Cache2-2} due to its large numbers of coupling constraints,
which make the first-order projected gradient descent method not applicable.
To get around the coupling constraints, we further establish Proposition~2, which exploits
the strong convexity of \eqref{strong convex}, so that the dual problem \eqref{dual p} can be efficiently solved by the proposed accelerated gradient based method.
The overall framework for solving the cache allocation problem
\eqref{Cache1} is depicted in Fig. \ref{f1-1}.}

\section{MCMB Design for Content Delivery}
To evaluate the performance of the proposed cache size allocation in Section III,
we further design the MCMB for content delivery
with fixed $\left\{C_k\right\}$.
Since only the multicast beamformers $\left\{\mathbf{V}_g\right\}$
are optimized for maximizing the instantaneous downloading sum-rate $R_{\text{sum}}$ in \eqref{DRsum}, the MCMB design problem becomes
\begin{subequations}\label{MCMB0}
\begin{equation}\label{MCMB0 obj}
\max_{\left\{\mathbf{V}_g\right\}}\sum_{g=1}^G\min_{k\in\mathcal{K}_g}
\left\{\frac{F_g}{F_g-C_k}\log\det\left(\mathbf{I}_N
+\mathbf{H}_k\mathbf{V}_{g}\mathbf{V}_{g}^H\mathbf{H}_k^H
\mathbf{J}_k\right)\right\},
\end{equation}
\begin{equation}\label{MCMB0 power}
\text{s.t.}\
\sum_{g=1}^G\mathrm{Tr}\left(\mathbf{V}_{g}
\mathbf{V}_{g}^H\right)\leq P_{\text{tot}},
\end{equation}
\end{subequations}
which can be equivalently transformed into the following smooth problem:
\begin{subequations}\label{MCMB1}
\begin{equation}\label{MCMB1 obj}
\min_{\left\{\mathbf{V}_g, \eta_g\right\}}-\sum_{g=1}^GF_g\eta_g,
\end{equation}
\begin{equation}\label{MCMB1 power}
\text{s.t.}\
\sum_{g=1}^G\mathrm{Tr}\left(\mathbf{V}_{g}
\mathbf{V}_{g}^H\right)\leq P_{\text{tot}},
\end{equation}
\begin{eqnarray}\label{MCMB1 rate}
(F_{g_k}-C_k)\eta_{g_k}
-\log\det\left(\mathbf{I}_N
+\mathbf{H}_k\mathbf{V}_{g_k}\mathbf{V}_{g_k}^H\mathbf{H}_k^H
\mathbf{J}_k\right)
\nonumber\\
\leq0,
\forall k=1, 2, \ldots, K.
\end{eqnarray}
\end{subequations}
Notice that problem \eqref{MCMB1} is a simplified version of \eqref{Cache2}
when $T=1$ and $\left\{C_k\right\}$ are fixed.
Following similar derivations
to that of the cache allocation problem in the last section (see Appendix \ref{derive MCMB}), the $i$-th SCA subproblem of \eqref{MCMB1} can be written as
\begin{subequations}\label{MCMB2}
\begin{equation}\label{MCMB2 obj}
\left\{\mathbf{V}_g^{(i+1)}, \eta_g^{(i+1)}\right\}=
\arg \min_{\left\{\mathbf{V}_g, \eta_g\right\}}-\sum_{g=1}^GF_g\eta_g,
\end{equation}
\begin{equation}\label{MCMB2 power}
\text{s.t.}\
\sum_{g=1}^G\mathrm{Tr}\left(\mathbf{V}_{g}
\mathbf{V}_{g}^H\right)\leq P_{\text{tot}}.
\end{equation}
\begin{equation}\label{MCMB2 rate}
~~~~\
(F_{g_k}-C_k)\eta_{g_k}+h_k^{(i)}\left(\left\{\mathbf{V}_g\right\}\right)\leq0,~~\forall k=1, 2, \ldots, K,
\end{equation}
\end{subequations}
where $h_k^{(i)}\left(\left\{\mathbf{V}_g\right\}\right)$ is given in \eqref{f} of Appendix \ref{derive MCMB}.
Since subproblem \eqref{MCMB2} is convex, it can be optimally solved by the
standard optimization toolbox based on the interior-point method (e.g., CVX \cite{Grant2013}).
Notice that different from \eqref{Cache2-2}, the SCA subproblem \eqref{MCMB2}
involves only a single channel realization, thus the interior-point method
would not induce a heavy computational burden.

\begin{algorithm}[t!]
\caption{Proposed Algorithm for Solving \eqref{MCMB0}}
\begin{algorithmic}[1]\footnotesize
\State Initialize $\mathbf{V}_{g}^{(0)}=\sqrt{\frac{P_{\text{tot}}}{GMd}}\mathbf{1}_{M\times d},~~\forall g=1,2,\ldots,G$.\\
$\textbf{repeat}$ ($i=0,1,\ldots$)\\
Compute $\mathbf{\hat{A}}_{k}^{(i)}$, $\mathbf{\hat{B}}_{k}^{(i)}$, and
$\hat{b}_{k}^{(i)}$ with \eqref{D}-\eqref{b},
$\forall k=1,2,\ldots,K$.\\
Solve the $i$-th SCA subproblem \eqref{MCMB2} with CVX.\\
$\textbf{until}$ convergence
\end{algorithmic}
\end{algorithm}

\textcolor{black}{
The overall framework for solving the MCMB design problem
\eqref{MCMB0} is depicted in Fig.~\ref{f2}, and the algorithm
is summarized in Algorithm 4.}
Without loss of generality, the initialization of Algorithm 4 is obtained by equally allocating the total transmit power $P_{\text{tot}}$ to each $\mathbf{V}_{g}^{(0)}$.
Since the constructed convex constraint in
\eqref{convex c1} satisfies properties $\textbf{C.a}$ and $\textbf{C.b}$ in Appendix \ref{derive MCMB},
Algorithm 4 is guaranteed to converge to a stationary point of problem \eqref{MCMB1} \cite{Razaviyayn2014}.

\iffalse
\textcolor{black}{
To better understand Algorithm~4, the
proposed algorithm framework is illustrated in Fig.~\ref{f2}.
Specifically, we first tackle the non-smoothness and the non-convexity of the original problem \eqref{MCMB0}. Since the resulting convexified problem
involves only a single channel realization and its problem size is not very large,
we further apply the interior-point method to solve it.}
\fi

\begin{figure}[t!]
	\begin{center}
        {\includegraphics[width=0.3\textwidth]{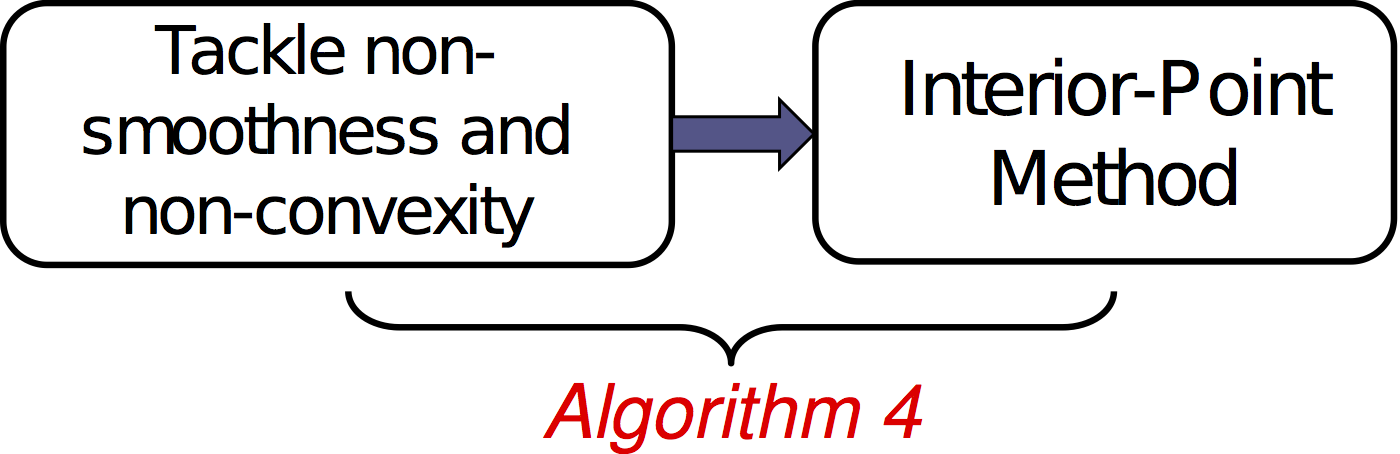}
		\caption{Proposed algorithm framework for the content delivery problem \eqref{MCMB0}.}\label{f2}}
	\end{center}
	\vspace{-0.2cm}
\end{figure}

\textcolor{black}{
\section{Caching Multiple Files with Different Popularities}
While the files were assumed to have the same popularity in Section III,
it can be extended to the more general case with multiple files having different popularities.
In particular, denote $f_g\in\mathcal{F}_g=
\left\{c_{g,1},c_{g,2},\ldots, c_{g,\left\vert
\mathcal{F}_g\right\vert}\right\}$ as a potential file to be requested by a BS group $g$
with the file's request probability $p_{f_g}$ (with $\sum_{f_g\in\mathcal{F}_g}{p_{f_g}}=1$).
Let $\mathbf{f}\triangleq\left[f_1, f_2, \ldots, f_G\right]^H$ and
$\mathcal{F}\triangleq\mathcal{F}_1\times\mathcal{F}_2\times\ldots\times\mathcal{F}_G$,
then the request probability of $\mathbf{f}$ is
$p_\mathbf{f}=\prod_{g=1}^Gp_{f_g}$ and we have $\sum_{\mathbf{f}\in\mathcal{F}}p_\mathbf{f}=1$.
Each BS $k$ caches $C_{k,f_{g_k}}/F_{f_{g_k}}$ of the file $f_{g_k}$.
Consequently, under the total cache size constraint $\sum_{k=1}^{K}\sum_{f_{g_k}\in\mathcal{F}_{g_k}}C_{k,f_{g_k}}\leq
C_{\text{tot}}$, the cache size allocation problem with different popularities
can be formulated as
\begin{subequations}\label{MULTI}
\begin{eqnarray}\label{MULTI1}
&&\max_{\left\{C_{k,f_{g}}\right\}, \left\{\mathbf{V}_{g,t,\mathbf{f}}\right\}}
\sum_{t=1}^T\sum_{g=1}^G\sum_{\mathbf{f}\in\mathcal{F}}p_\mathbf{f}\min_{k\in\mathcal{K}_g}
\Bigg\{\frac{F_{f_g}}{F_{f_g}-C_{k,f_{g}}}
\nonumber\\
&&\cdot\log\det\left(\mathbf{I}_N
+\mathbf{H}_{k,t,\mathbf{f}}\mathbf{V}_{g,t,\mathbf{f}}\mathbf{V}_{g,t,\mathbf{f}}^H\mathbf{H}_{k,t,\mathbf{f}}^H
\mathbf{J}_{k,t,\mathbf{f}}\right)\Bigg\},~~~~~
\end{eqnarray}
\begin{equation}\label{MULTI2}
\text{s.t.}\
\sum_{g=1}^G\mathrm{Tr}\left(\mathbf{V}_{g,t,\mathbf{f}}
\mathbf{V}_{g,t,\mathbf{f}}^H\right)\leq P_{\text{tot}},~~\forall \mathbf{f}\in\mathcal{F},~~\forall t=1, 2, \ldots, T,
\end{equation}
\begin{equation}\label{MULTI3}
~~~~\
\sum_{k=1}^{K}\sum_{f_{g_k}\in\mathcal{F}_{g_k}}C_{k,f_{g_k}}\leq
C_{\text{tot}},~~~~~~~~~~~~~~~~~~~~~~~~~~~~~~~~~~~~~~
\end{equation}
\begin{equation}\label{MULTI4}
~~~~\
0\leq C_{k,f_{g_k}} \leq F_{f_{g_k}},~~\forall f_{g_k}\in\mathcal{F}_{g_k},~~\forall k=1, 2, \ldots, K,
\end{equation}
\end{subequations}
where
$\left\{\mathbf{H}_{k,t,\mathbf{f}}\right\}_{k=1}^K$ and $\left\{\mathbf{V}_{g,t,\mathbf{f}}\right\}_{g=1}^G$ are the corresponding channel matrices
and beamformers, respectively.
Notice that
when $p_\mathbf{f}$ with different $\mathbf{f}$ are equal,
\eqref{MULTI} will reduce to \eqref{Cache1}.
With the only difference between \eqref{MULTI} and \eqref{Cache1}
lying in the weighting parameter $p_\mathbf{f}$ in~\eqref{MULTI1},
\eqref{MULTI} can be solved using the same
algorithm framework as Algorithm 3 proposed in Section III.
The corresponding algorithm is summarized as Algorithm~5 in Appendix~\ref{al5}.}
\textcolor{black}{
Finally, notice that an unpopular file would lead to a small $C_{k,f_{g_k}}$.
But no matter $C_{k,f_{g_k}}$ is big or small, for content delivery, as shown in Section IV, whenever a file (even an unpopular file) is requested by the BSs, the computation center must provide multicast transmission immediately. Therefore, there is no delay during the transmission.}

\section{Simulation Results and Discussions}
In this section, we evaluate the performance of the proposed cache size allocation
in terms of downloading sum-rate through simulations.
All simulations are performed on MATLAB R2017b running on a Windows x64 machine with 3.3 GHz CPU and 8 GB RAM.

We consider a downlink C-RAN with $G=4$ clusters of BSs, where each cluster consists of
$3$ BSs. In particular, the distances between the computation center and the BSs
in the $4$ clusters are $\left\{160,260,360\right\}$,
$\left\{200, 280, 360\right\}$,
$\left\{160,280,400\right\}$, and
$\left\{240,320,400\right\}$ in meters, respectively.
The computation center is equipped with $M=20$ antennas and each BS is equipped with $N=2$ antennas.
The path loss at distance $D$ kilometers follows $128.1+37.6\log_{10}(D)$ in dB,
and the small-scale channel is subject to Rayleigh fading.
We use $T=100$ sets of channel realizations
for solving the cache allocation problem~\eqref{Cache1}.
The maximum transmit power at the computation center is $P_{\text{tot}}=40$ W
and the antenna gain is $17$ dBi \cite{Dai2018}. The backhaul channel bandwidth is $20$ MHz and the background noise power spectral density is $-150$ dBm/Hz. The content size is $F_g=100, \forall g=1,2,3,4$, and the budget of the total cache size is $C_{\text{tot}}=120$.

The step size $\beta_s$ in Algorithm 1 and Algorithm 2 is fixed as $1$,
and the parameters in \eqref{strong convex} are fixed as $\rho_1=10^5$, $\rho_2=10^4$, and $\rho_3=1$.
The iterations of Algorithm 1, Algorithm 2, and Algorithm 4
terminate when the relative changes of the corresponding objective functions
between two consecutive iterations are less than $10^{-3}$.
The SCA iteration of Algorithm 3 stops when the relative decrease of the cost function \eqref{Cache2 obj}
in the last $100$ iterations is less than $10^{-2}$.

\subsection{Convergence Behaviors of Proposed Algorithms}

\begin{figure}[t!]
\begin{center}
 \subfigure[]{
  \includegraphics[width=0.35\textwidth]{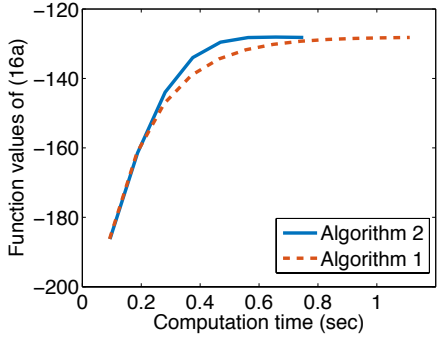}\label{inner}
  }
  \subfigure[]{
  \includegraphics[width=0.35\textwidth]{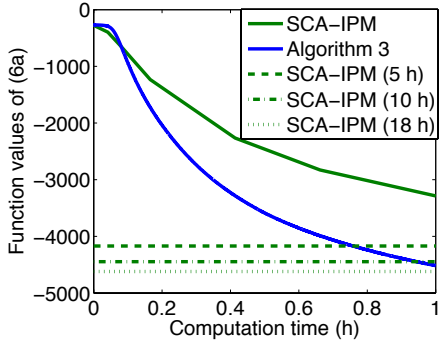}\label{outer}
}\caption{(a) Convergence behavior of the inner iteration of Algorithm 3.
(b) Convergence behavior of the outer iteration of Algorithm 3.}
\end{center}
\vspace{-0.2cm}
\end{figure}

First, we show the convergence behavior of Algorithm 3 for solving the
cache size allocation problem \eqref{Cache1}.
Since the inner iteration of Algorithm 3 is based on Algorithm 1 or Algorithm 2, we illustrate the convergence behaviors of Algorithm 1 and Algorithm 2 in Fig. \ref{inner}.
It can be seen that
they both converge to the same objective function value, but
Algorithm 2 achieves much faster convergence than Algorithm 1, with computation time roughly
half of that of Algorithm 1.
This demonstrates the effectiveness of the acceleration technique exploited in Algorithm 2.
Therefore, we only adopt Algorithm 2 as the inner iteration of Algorithm 3 for the rest of simulations.

On the other hand, the convergence behavior of the outer iteration of Algorithm 3 is shown in Fig. \ref{outer}. To show the complexity advantage of Algorithm 3, we compare it with a second-order algorithm (termed as SCA-IPM), which also applies the SCA framework but solves each SCA subproblem \eqref{Cache2-2} with the interior-point method.
It can be seen that with the proposed algorithm and SCA-IPM running for the same amount of time, the proposed algorithm achieves a much lower function value of \eqref{Cache2 obj}.
Furthermore, in order to achieve the same accuracy, SCM-IPM requires 5 to 18 times more computation time.
This demonstrates that Algorithm 3 is more suitable for the cache size allocation problem \eqref{Cache1}
with a large number of variables and constraints.
Notice that the computation time is measured based on MATLAB implementation. In
industrial implementation, more efficient programming languages (e.g., C++)
would be adopted to achieve even more efficient execution.

\subsection{Performance of Proposed Algorithms on Downloading Sum-Rate}
In this subsection, we demonstrate the effectiveness of the proposed cache size allocation
in terms of downloading sum-rate (with MCMB design using Algorithm 4).
The sum-rate achieved by the proposed cache allocation is compared to that of the uniform cache size allocation.
\textcolor{black}{
Furthermore, since the recent work \cite{Dai2018} is designed for the single-cluster multicast scenario
without inter-cluster interference,
we compare with a benchmark by
extending \cite{Dai2018} to the multi-cluster multicast scenario in a time-division multiplexing manner.
Specifically, for a $G$-cluster multicast scenario, to avoid the inter-cluster interference,
each cluster utilizes $1/G$ of the total time resources. Consequently,
without inter-cluster interference, the cache allocation problem
can be directly solved by the proposed algorithm in \cite{Dai2018}.
To show the importance of modeling the inter-cluster interference, we also add
a baseline, which directly applies \cite{Dai2018} without modeling the interference for optimization.}

Figure \ref{COM} illustrates the cumulative distribution functions of downloading sum-rates obtained by different schemes under $400$ channel realizations.
It can be seen that, by simultaneously delivering the required contents to different multicast clusters, the multi-cluster transmission schemes achieve significantly higher downloading sum-rate than the time-division single-cluster transmission schemes.
This demonstrates the necessity of the multi-cluster multicast transmission for improving the backhaul efficiency.
Moreover, with the optimized cache size, the downloading sum-rate achieved by the proposed algorithm is much higher than that of the uniform caching scheme.
\textcolor{black}{On the other hand,
without effective interference management,
the direct application of [8] results in the lowest downloading sum-rate.
This demonstrates the necessity of modeling the interference for the multi-cluster multicast scenario.}
\textcolor{black}{
The comparison among various schemes is summarized in Table I.
It can be seen that the average downloading sum-rate achieved by the proposed approach
is more than 2 times of that of the extension of \cite{Dai2018},
and more than 10 times of that of the direct application of \cite{Dai2018}, respectively.}

\begin{figure}[t!]
\begin{center}
  \includegraphics[width=0.35\textwidth]{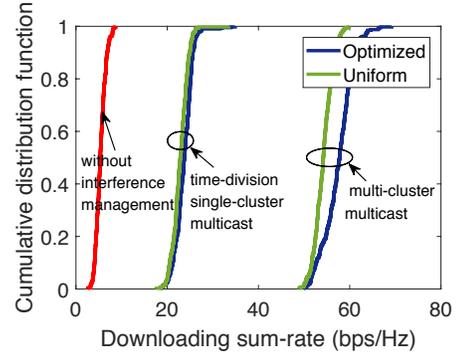}
  \caption{Cumulative distribution functions of downloading sum-rates obtained by different schemes under $400$ channel realizations.}\label{COM}
\end{center}
\vspace{-0.2cm}
\end{figure}

\begin{table*}[t]\caption{Average downloading sum-rate comparison}
\begin{center}\footnotesize
\begin{tabular}{ccccc}
\hline
 Schemes     & The way to handle multi-cluster interference  &  Average downloading sum-rate \\ \hline
 Proposed approach     & Multi-cluster beamforming  &  $59$ bps/Hz \\ \hline
 Direct application of [8] & No mitigation  &   $5$ bps/Hz \\  \hline
 Extending [8] to multi-cluster scenario & Time-division multiplexing &  $26$ bps/Hz \\ \hline
\end{tabular}
\end{center}
\vspace{-0.2cm}
\end{table*}

To reveal the insight of the superiority of the proposed cache size allocation, we show the
optimized cache sizes among different BSs in Fig. \ref{Cache-COM}.
It can be seen that, within a cluster, the cache size allocated by the proposed algorithm
increases with the distance between the BS and the computation center.
For instance, in Cluster 1, BS 3 is allocated with the largest cache sizes.
On the other hand, among different clusters,
the proposed algorithm allocates larger cache sizes to the nearer clusters.
For example, Cluster 1 and Cluster 4 are allocated with the the largest
and the smallest cache sizes, respectively.

\begin{figure}[t!]
\begin{center}
  \includegraphics[width=0.35\textwidth]{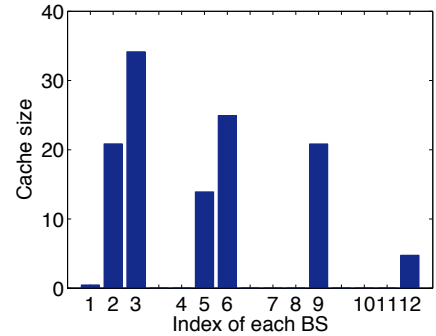}
  \caption{Cache size allocation results among different BSs. The distances between the computation center and the BSs
in the $4$ clusters are $\left\{160,260,360\right\}$,
$\left\{200, 280, 360\right\}$,
$\left\{160,280,400\right\}$, and
$\left\{240,320,400\right\}$ in meters.}\label{Cache-COM}
\end{center}
\vspace{-0.2cm}
\end{figure}

\textcolor{black}{
Finally, to investigate the impact of file popularity,
the allocated cache sizes among different files are shown in Fig.~\ref{BAR-3},
where there are $2$ clusters of BSs, and each cluster consists of $3$ BSs
situated in $160$, $260$, and $360$ meters from the computation center.
In Fig.~\ref{BAR55-3} to Fig.~\ref{BAR91-3}, each BS requests $2$ files with popularities $(p_1=0.5, p_2=0.5)$,
$(p_1=0.7, p_2=0.3)$, and $(p_1=0.9, p_2=0.1)$, respectively.
It can be seen that the weakest BS $3$ and BS $6$ are always allocated
the largest cache sizes in all the three cases.
Furthermore,
as the difference between the popularities of the two files increases,
the file with higher popularity is allocated much larger cache size.
For instance, as shown in Fig.~\ref{BAR91-3}, when $p_1=0.9$ and $p_2=0.1$,
the file with higher popularity
is allocated almost all the cache.}

\begin{figure}[t!]
	\begin{center}
		\subfigure[$p_1=0.5, p_2=0.5$]{\includegraphics[width=0.35\textwidth]{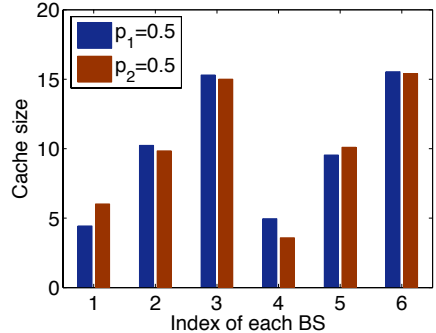} \label{BAR55-3}}
		\subfigure[$p_1=0.7, p_2=0.3$]{\includegraphics[width=0.35\textwidth]{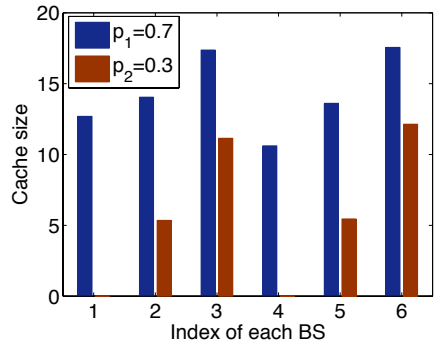} \label{BAR73-3}}
        \subfigure[$p_1=0.9, p_2=0.1$]{\includegraphics[width=0.35\textwidth]{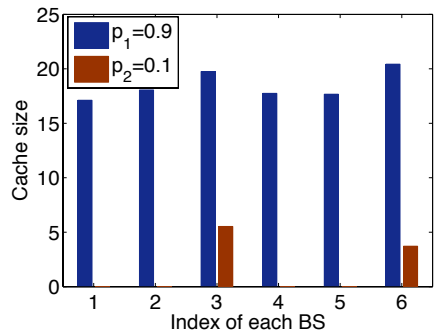} \label{BAR91-3}}
		\caption{Cache size allocation results among different files with different popularities.} \label{BAR-3}
	\end{center}
	\vspace{-0.2cm}
\end{figure}

\section{Conclusions}
Caching at BSs was studied for a C-RAN with multi-cluster multicast backhaul,
with the aim of maximizing the content downloading sum-rate of the wireless backhaul
under a total cache budget constraint.
To solve this large-scale nonsmooth nonconvex problem, we proposed
an accelerated first-order algorithm,
which achieves much higher content downloading sum-rate than a second-order algorithm running for the same amount of time. Moreover, with multi-cluster multicast transmission,
the proposed algorithm achieves significantly higher downloading sum-rate than those of
time-division single-cluster transmission schemes. In addition,
simulation results revealed that the proposed algorithm
allocates larger cache sizes to farther BSs within nearer clusters,
which provides insight to the superiority of the proposed cache allocation.

\numberwithin{equation}{section}
\appendices

\section{Proof of Proposition \ref{lm2}}\label{proof of lm2}
From \eqref{f2}, we can rewrite $f_{k,t}^{(i)}\left(C_{k}, \eta_{g_k,t}, \left\{\mathbf{V}_{g,t}\right\}_{g=1}^G\right)=\phi_{k,t}^{(i)}\left(\left\{\mathbf{V}_{g,t}\right\}_{g=1}^G\right)+\psi_{k,t}^{(i)}\left(C_{k}, \eta_{g_k,t}\right)$, where
\begin{eqnarray}\label{phi}
&&\phi_{k,t}^{(i)}\left(\left\{\mathbf{V}_{g,t}\right\}_{g=1}^G\right)
=\sum_{g=1}^G\mathrm{Tr}\left(\mathbf{V}_{g,t}^H
\mathbf{A}_{k,t}^{(i)}\mathbf{V}_{g,t}\right)
\nonumber\\
&&+2\Re\left\{\mathrm{Tr}\left(\mathbf{B}_{k,t}^{(i)}\mathbf{V}_{g_k,t}\right)\right\}
+b_{k,t}^{(i)}-\frac{\left(\eta_{g_k,t}^{(i)}+C_k^{(i)}\right)^2}{2},~~~~~~
\end{eqnarray}
\begin{eqnarray}\label{psi}
\psi_{k,t}^{(i)}\left(C_{k}, \eta_{g_k,t}\right)
=\frac{\eta_{g_k,t}^2+C_k^2}{2}
-\left(\eta_{g_k,t}^{(i)}+C_k^{(i)}\right)
\nonumber\\
\cdot\left(\eta_{g_k,t}+C_k\right)+F_{g_k}\eta_{g_k,t}+\frac{\left(\eta_{g_k,t}^{(i)}+C_k^{(i)}\right)^2}{2}.
\end{eqnarray}
From \eqref{psi}, we have
\begin{eqnarray}\label{psi2}
&&\psi_{k,t}^{(i)}\left(C_{k}, \eta_{g_k,t}\right)-(F_{g_k}-C_k)\eta_{g_k,t}
\nonumber\\
&=&\frac{1}{2}\left(\eta_{g_k,t}+C_k-\eta_{g_k,t}^{(i)}-C_k^{(i)}\right)^2\geq0,
\end{eqnarray}
where the equality holds at $C_k=C_k^{(i)}$ and $\eta_{g_k}=\eta_{g_k}^{(i)}$.

Next, we prove
\begin{eqnarray}\label{phi2}
&&\phi_{k,t}^{(i)}\left(\left\{\mathbf{V}_{g,t}\right\}_{g=1}^G\right)
\nonumber\\
&\geq&-\log\det\left(\mathbf{I}_N
+\mathbf{H}_{k,t}\mathbf{V}_{g_k,t}\mathbf{V}_{g_k,t}^H\mathbf{H}_{k,t}^H
\mathbf{J}_{k,t}\right).~~~~
\end{eqnarray}
More specifically, applying $\det\left(\mathbf{I}+\mathbf{X}\mathbf{Y}\right)=\det\left(\mathbf{I}+\mathbf{Y}\mathbf{X}\right)$,
we have
\begin{eqnarray}\label{det}
&&\log\det\left(\mathbf{I}_N
+\mathbf{H}_{k,t}\mathbf{V}_{g_k,t}\mathbf{V}_{g_k,t}^H\mathbf{H}_{k,t}^H
\mathbf{J}_{k,t}\right)
\nonumber\\
&=&
\log\det\left(\mathbf{I}_d
+\mathbf{V}_{g_k,t}^H\mathbf{H}_{k,t}^H
\mathbf{J}_{k,t}\mathbf{H}_{k,t}\mathbf{V}_{g_k,t}\right).
\end{eqnarray}
Applying the Woodbury matrix identity to the right-hand-side of \eqref{det}, we have
\begin{eqnarray}\label{det2}
&&\log\det\left(\mathbf{I}_N
+\mathbf{H}_{k,t}\mathbf{V}_{g_k,t}\mathbf{V}_{g_k,t}^H\mathbf{H}_{k,t}^H
\mathbf{J}_{k,t}\right)
\nonumber\\
&=&\log\det\left(
{\left(\underbrace{\mathbf{I}_d-
\mathbf{U}_{k,t}^H\mathbf{H}_{k,t}\mathbf{V}_{g_k,t}}_
{\mathbf{Q}_{k,t}}\right)}^{-1}\right),
\end{eqnarray}
where $\mathbf{U}_{k,t}$ is defined as
\begin{equation}\label{Ut}
\mathbf{U}_{k,t}\triangleq
\left(\sum_{g=1}^G\mathbf{H}_{k,t}\mathbf{V}_{g,t}
\mathbf{V}_{g,t}^H\mathbf{H}_{k,t}^H
+\sigma_k^2\mathbf{I}_{N}\right)^{-1}\mathbf{H}_{k,t}\mathbf{V}_{g_k,t}.
\end{equation}
Notice that the right-hand-side of \eqref{det2} is convex over ${\mathbf{Q}_{k,t}}$, thus
by using first-order Taylor expansion at $\mathbf{Q}_{k,t}^{(i)}\triangleq
\mathbf{I}_{d}-\left(\mathbf{U}_{k,t}^{(i)}\right)^H\mathbf{H}_{k,t}\mathbf{V}_{g_k,t}^{(i)}$, we have
\begin{eqnarray}\label{first-order expansion}
&&\log\det\left(\mathbf{I}_N
+\mathbf{H}_{k,t}\mathbf{V}_{g_k,t}\mathbf{V}_{g_k,t}^H\mathbf{H}_{k,t}^H
\mathbf{J}_{k,t}\right)
\nonumber\\
&\geq&
\log\det\left(\left(\mathbf{Q}_{k,t}^{(i)}\right)^{-1}\right)
-\mathrm{Tr}\left(\left(\mathbf{Q}_{k,t}^{(i)}\right)^{-1}\mathbf{Q}_{k,t}\right)
+d.\nonumber\\
\end{eqnarray}
Furthermore, we majorize $\mathbf{Q}_{k,t}$ by
\begin{equation}\label{E}
\mathbf{E}_{k,t}=
\mathbf{Q}_{k,t}
+\left(\mathbf{U}_{k,t}^{(i)}-\mathbf{U}_{k,t}\right)^H
\mathbf{X}_{k,t}
\left(\mathbf{U}_{k,t}^{(i)}-\mathbf{U}_{k,t}\right),
\end{equation}
where $\mathbf{X}_{k,t}\triangleq\sum_{g=1}^G\mathbf{H}_{k,t}\mathbf{V}_{g,t}
\mathbf{V}_{g,t}^H\mathbf{H}_{k,t}^H
+\sigma_k^2\mathbf{I}_{N}$.
Substituting \eqref{E} into \eqref{first-order expansion} yields
\begin{eqnarray}\label{ine}
&&-\log\det\left(\mathbf{I}_N
+\mathbf{H}_{k,t}\mathbf{V}_{g_k,t}\mathbf{V}_{g_k,t}^H\mathbf{H}_{k,t}^H
\mathbf{J}_{k,t}\right)
\nonumber\\
&\leq&\mathrm{Tr}\left(\left(\mathbf{Q}_{k,t}^{(i)}\right)^{-1}\mathbf{E}_{k,t}
\right)-\log\det\left(\left(\mathbf{Q}_{k,t}^{(i)}\right)^{-1}\right)
-d.\nonumber\\
\end{eqnarray}
Then we show the right-hand-side of \eqref{ine} is equal to $\phi_{k,t}^{(i)}\left(\left\{\mathbf{V}_{g,t}\right\}_{g=1}^G\right)$.
Substituting the expressions of $\mathbf{Q}_{k,t}$ in \eqref{det2}
and $\mathbf{U}_{k,t}$ in \eqref{Ut} into $\mathbf{E}_{k,t}$, we have
\begin{equation}\label{E2}
\mathbf{E}_{k,t}=
\left(\mathbf{U}_{k,t}^{(i)}\right)^H\mathbf{X}_{k,t}\mathbf{U}_{k,t}^{(i)}
-2\Re\left\{
\left(\mathbf{U}_{k,t}^{(i)}\right)^H\mathbf{H}_{k,t}\mathbf{V}_{g_k,t}\right\}+
\mathbf{I}_d.
\end{equation}
By applying the expressions of $\mathbf{E}_{k,t}$,
$\mathbf{A}_{k,t}^{(i)}$, $\mathbf{B}_{k,t}^{(i)}$, and $b_{k,t}^{(i)}$ in
\eqref{E2}, \eqref{Dt}, \eqref{Ct}, and \eqref{bt},
the right-hand-side of \eqref{ine} is equal to $\phi_{k,t}^{(i)}\left(\left\{\mathbf{V}_{g,t}\right\}_{g=1}^G\right)$,
thus \eqref{phi2} is proved.

Now, we show the equality in \eqref{phi2} holds at $\left\{\mathbf{V}_{g,t}\right\}_{g=1}^G=\left\{\mathbf{V}_{g,t}^{(i)}\right\}_{g=1}^G$.
When $\left\{\mathbf{V}_{g,t}\right\}_{g=1}^G=\left\{\mathbf{V}_{g,t}^{(i)}\right\}_{g=1}^G$,
we have $\mathbf{U}_{k,t}=\mathbf{U}_{k,t}^{(i)}$ and
$\mathbf{Q}_{k,t}=\mathbf{Q}_{k,t}^{(i)}$,
thus the equality in \eqref{first-order expansion} holds.
Moreover, from \eqref{E}, seeing that
$\mathbf{Q}_{k,t}=\mathbf{E}_{k,t}$ at $\mathbf{U}_{k,t}=\mathbf{U}_{k,t}^{(i)}$,
the equality in \eqref{ine} also holds.
Since we have shown that the right-hand-side of \eqref{ine} is equal to $\phi_{k,t}^{(i)}\left(\left\{\mathbf{V}_{g,t}\right\}_{g=1}^G\right)$,
the equality in \eqref{phi2} also holds at $\left\{\mathbf{V}_{g,t}\right\}_{g=1}^G=\left\{\mathbf{V}_{g,t}^{(i)}\right\}_{g=1}^G$.
Therefore, adding up the left-hand-sides of \eqref{psi2} and \eqref{phi2},
we complete the proof for \textbf{(1.1)} of Proposition \ref{lm2}.

Finally, we prove \textbf{(1.2)} of Proposition \ref{lm2}.
From \eqref{psi} we have
\begin{equation}\label{psi3}
\frac{\partial\psi_{k,t}^{(i)}\left(C_{k}^{(i)}, \eta_{g_k,t}^{(i)}\right)}{\partial C_k}=-\eta_{g_k,t}^{(i)},
\end{equation}
\begin{equation}\label{psi4}
\frac{\partial\psi_{k,t}^{(i)}\left(C_{k}^{(i)}, \eta_{g_k,t}^{(i)}\right)}{\partial \eta_{g_k,t}}=
F_{g_k}-C_k^{(i)}.
\end{equation}
On the other hand, since \eqref{first-order expansion} is the first-order expansion of
$\log\det\left(\mathbf{I}_N
+\mathbf{H}_{k,t}\mathbf{V}_{g_k,t}\mathbf{V}_{g_k,t}^H\mathbf{H}_{k,t}^H
\mathbf{J}_{k,t}\right)$ at $\mathbf{Q}_{k,t}=\mathbf{Q}_{k,t}^{(i)}$,
we have
\begin{eqnarray}\label{dR}
&&\frac{\partial}{\partial v}\left(\log\det\left(\mathbf{I}_N
+\mathbf{H}_{k,t}\mathbf{V}_{g_{k,t}}^{(i)}\left(\mathbf{V}_{g_{k,t}}^{(i)}\right)^H\mathbf{H}_{k,t}^H
\mathbf{J}_{k,t}^{(i)}\right)\right)
\nonumber\\
&=&
-\mathrm{Tr}\left(\left(\mathbf{Q}_{k,t}^{(i)}\right)^{-1}
\frac{\partial\mathbf{Q}_{k,t}^{(i)}}{\partial v}\right),
\end{eqnarray}
where $v$ represents any element of $\mathbf{V}_{g,t}$, $\forall g=1,2,\ldots,G$.
From \eqref{E}, we have
\begin{eqnarray}\label{dE}
\frac{\partial\mathbf{E}_{k,t}^{(i)}}{\partial v}
&=&\frac{\partial\mathbf{Q}_{k,t}^{(i)}}{\partial v}
+\frac{\partial}{\partial v}
\Bigg(\left(\mathbf{U}_{k,t}^{(i)}-\mathbf{U}_{k,t}\right)^H
\nonumber\\
&&\cdot\mathbf{X}_{k,t} \left.\Bigg)\right\vert_{\big\{\mathbf{V}_{g,t}\big\}_{g=1}^G=\big\{\mathbf{V}_{g,t}^{(i)}\big\}_{g=1}^G}
\left(\mathbf{U}_{k,t}^{(i)}-\mathbf{U}_{k,t}^{(i)}\right)
\nonumber\\
&&+\left(\mathbf{U}_{k,t}^{(i)}-\mathbf{U}_{k,t}^{(i)}\right)^H
\mathbf{X}_{k,t}
\nonumber\\
&&\left.\frac{\partial}{\partial v}\left(\mathbf{U}_{k,t}^{(i)}-\mathbf{U}_{k,t}\right)\right\vert_{\big\{\mathbf{V}_{g,t}\big\}_{g=1}^G=\big\{\mathbf{V}_{g,t}^{(i)}\big\}_{g=1}^G}
\nonumber\\
&=&\frac{\partial\mathbf{Q}_{k,t}^{(i)}}{\partial v}.
\end{eqnarray}
Substituting \eqref{dE} into \eqref{dR} yields
\begin{eqnarray}\label{dR2}
&&\frac{\partial}{\partial v}\left(\log\det\left(\mathbf{I}_N
+\mathbf{H}_{k,t}\mathbf{V}_{g_{k,t}}^{(i)}\left(\mathbf{V}_{g_{k,t}}^{(i)}\right)^H\mathbf{H}_{k,t}^H
\mathbf{J}_{k,t}^{(i)}\right)\right)
\nonumber\\
&=&
-\mathrm{Tr}\left(\left(\mathbf{Q}_{k,t}^{(i)}\right)^{-1}
\frac{\partial\mathbf{E}_{k,t}^{(i)}}{\partial v}\right)
=-\frac{\partial }{\partial v}\phi_{k,t}^{(i)}\left(\left\{\mathbf{V}_{g,t}\right\}_{g=1}^G\right),\nonumber\\
\end{eqnarray}
where the second equality follows from the fact that the right-hand-side of \eqref{ine} is equal to $\phi_{k,t}^{(i)}\left(\left\{\mathbf{V}_{g,t}\right\}_{g=1}^G\right)$.
Combining \eqref{psi3}, \eqref{psi4}, and \eqref{dR2}, we complete the proof for
\textbf{(1.2)} of Proposition \ref{lm2}.

\section{Proof of Proposition \ref{smooth}}\label{proof of smooth}
The partial Lagrangian function of \eqref{Cache3} is
\begin{eqnarray}\label{LR}
&&L\left(\left\{C_{k}\right\}, \left\{\mathbf{V}_{g,t}, \eta_{g,t}\right\}, \left\{\delta_{t}\right\}, \left\{\lambda_{k,t}\right\}, \mu\right)\nonumber\\
&=&\Upsilon^{(i)}\left(\left\{C_{k}\right\}, \left\{\mathbf{V}_{g,t}, \eta_{g,t}\right\}\right)
\nonumber\\
&&+\sum_{t=1}^T\delta_t\left(\sum_{g=1}^G\mathrm{Tr}\left(\mathbf{V}_{g,t}
\mathbf{V}_{g,t}^H\right)-P_{\text{tot}}\right)
\nonumber\\&&+\sum_{t=1}^T\sum_{k=1}^K
\lambda_{k,t}f_{k,t}^{(i)}\left(C_{k}, \eta_{g_k,t}, \left\{\mathbf{V}_{g,t}\right\}_{g=1}^G\right)
\nonumber\\
&&+\mu\left(\sum_{k=1}^KC_k-C_{\text{tot}}\right),
\end{eqnarray}
where $\left\{\delta_{t}\right\}$, $\mu$, and $\left\{\lambda_{k,t}\right\}$
are the non-negative dual variables corresponding to the coupling constraints \eqref{Cache3 power}, \eqref{Cache3 cache1}, and \eqref{Cache3 rate}. Consequently,
the dual function of \eqref{Cache3} is defined as
\begin{eqnarray}\label{Cache4}
&&D\left(\left\{\delta_{t}\right\}, \left\{\lambda_{k,t}\right\}, \mu\right)
\triangleq
\nonumber\\
&&\min_{\left\{C_{k}\right\}, \left\{\mathbf{V}_{g,t}, \eta_{g,t}\right\}}L\left(\left\{C_{k}\right\}, \left\{\mathbf{V}_{g,t}, \eta_{g,t}\right\}, \left\{\delta_{t}\right\}, \left\{\lambda_{k,t}\right\}, \mu\right),
\nonumber\\
&&\text{s.t.}\ \eqref{Cache3 cache2}.
\end{eqnarray}
By substituting the expressions of
$f_{k,t}^{(i)}\left(C_{k}, \eta_{g_k,t}, \left\{\mathbf{V}_{g,t}\right\}_{g=1}^G\right)$
in \eqref{f2} and $\Upsilon^{(i)}\left(\left\{C_{k}\right\}, \left\{\mathbf{V}_{g,t}, \eta_{g,t}\right\}\right)$
in \eqref{strong convex} into \eqref{Cache4}, problem \eqref{Cache4} becomes
\begin{subequations}\label{Cache5}
\begin{eqnarray}\label{Cache5 obj}
&&\min_{\left\{C_{k}\right\},\left\{\mathbf{V}_{g,t}, \eta_{g,t}\right\}}-\sum_{t=1}^T\sum_{g=1}^GF_g\eta_{g,t}+\frac{\rho_3}{2}\sum_{k=1}^K\left(C_k-C_k^{(i)}\right)^2
\nonumber\\&&
+\frac{\rho_1}{2}\sum_{t=1}^T\sum_{g=1}^G\left(\eta_{g,t}-\eta_{g,t}^{(i)}\right)^2
\nonumber\\&&
+\rho_2\sum_{t=1}^T\sum_{g=1}^G\left\Vert\mathbf{V}_{g,t}-\mathbf{V}_{g,t}^{(i)}\right\Vert_F^2
+\mu\left(\sum_{k=1}^KC_k-C_{\text{tot}}\right)\nonumber\\
&&
+\sum_{t=1}^T\delta_t\left(\sum_{g=1}^G\mathrm{Tr}\left(\mathbf{V}_{g,t}
\mathbf{V}_{g,t}^H\right)-P_{\text{tot}}\right)
\nonumber\\&&+\sum_{t=1}^T\sum_{k=1}^K
\lambda_{k,t}\Bigg(
\sum_{g=1}^G\mathrm{Tr}\left(\mathbf{V}_{g,t}^H
\mathbf{A}_{k,t}^{(i)}\mathbf{V}_{g,t}\right)
\nonumber\\
&&
+2\Re\left\{\mathrm{Tr}\left(\mathbf{B}_{k,t}^{(i)}\mathbf{V}_{g_k,t}\right)\right\}
+\frac{\eta_{g_k,t}^2+C_k^2}{2}+F_{g_k}\eta_{g_k,t}\nonumber\\
&&-\left(\eta_{g_k,t}^{(i)}+C_k^{(i)}\right)\left(\eta_{g_k,t}+C_k\right)
+b_{k,t}^{(i)}\Bigg),
\end{eqnarray}
\begin{equation}\label{Cache5 cache2}
\text{s.t.}\
0\leq C_k \leq F_{g_k},~~\forall k=1, 2, \ldots, K,
\end{equation}
\end{subequations}
Due to the variable separability of \eqref{Cache5 obj}, problem \eqref{Cache5} can be decomposed into
$2GT+K$ subproblems in parallel. Specifically,
there are $GT$ subproblems over $\left\{\eta_{g,t}\right\}$, with each written as
\begin{eqnarray}\label{Cache5-1}
\min_{\eta_{g,t}}&&-F_g\eta_{g,t}+\frac{\rho_1}{2}\left(\eta_{g,t}-\eta_{g,t}^{(i)}\right)^2
\nonumber\\
&&+\sum_{k\in\mathcal{K}_g}\lambda_{k,t}\left(\frac{\eta_{g,t}^2}{2}
+\left(F_g-\eta_{g,t}^{(i)}-C_k^{(i)}\right)\eta_{g,t}\right).~~~~~~
\end{eqnarray}
Since the cost function \eqref{Cache5-1} is strongly convex over $\eta_{g,t}$, by setting its gradient to zero, the minimizer
$\eta_{g,t}^{\diamond}$ is uniquely given by \eqref{eta}.
Moreover, there are $GT$ subproblems over $\left\{\mathbf{V}_{g,t}\right\}$, with each written as
\begin{eqnarray}\label{Cache5-2}
\min_{\mathbf{V}_{g,t}}
&&\rho_2\left\Vert\mathbf{V}_{g,t}-\mathbf{V}_{g,t}^{(i)}\right\Vert_F^2
+\sum_{k=1}^K\lambda_{k,t}\mathrm{Tr}\left(\mathbf{V}_{g,t}^H
\mathbf{A}_{k,t}^{(i)}\mathbf{V}_{g,t}\right)
\nonumber\\
&&
+2\sum_{k\in\mathcal{K}_g}\lambda_{k,t}\Re\left\{\mathrm{Tr}\left(\mathbf{B}_{k,t}^{(i)}\mathbf{V}_{g,t}\right)\right\}
+\delta_t\mathrm{Tr}\left(\mathbf{V}_{g,t}
\mathbf{V}_{g,t}^H\right).\nonumber\\
\end{eqnarray}
Since the cost function \eqref{Cache5-2} is strongly convex over $\mathbf{V}_{g,t}$, by setting its gradient to zero, the minimizer $\mathbf{V}_{g,t}^{\diamond}$ is uniquely given by \eqref{V}.
In addition, there are $K$ subproblems over $\left\{C_k\right\}$, with each written as
\begin{eqnarray}\label{Cache5-3}
\min_{0\leq C_k \leq F_{g_k}}&&\mu C_k+\frac{\rho_3}{2}\left(C_k-C_k^{(i)}\right)^2
\nonumber\\&&
+\sum_{t=1}^T\lambda_{k,t}\left(\frac{C_k^2}{2}-\left(\eta_{g_k,t}^{(i)}+C_k^{(i)}\right)C_k\right).~~~~~
\end{eqnarray}
Since the cost function \eqref{Cache5-3} is strongly convex quadratic over the scalar variable $C_k$,
by setting its gradient to zero and then projecting the solution to $0\leq C_k \leq F_{g_k}$,
the minimizer $C_k^{\diamond}$ is uniquely given by \eqref{Ck}.
By substituting the optimal solution $\left\{\mathbf{V}_{g,t}^{\diamond}, \eta_{g,t}^{\diamond}\right\}$
and $\left\{C_{k}^{\diamond}\right\}$ into \eqref{Cache4}, the dual function $D\left(\left\{\delta_{t}\right\}, \left\{\lambda_{k,t}\right\}, \mu\right)$ is expressed in a closed-form as shown in \eqref{dual}.

\section{SCA Subproblem of \eqref{MCMB0}}\label{derive MCMB}
To tackle the non-convexity of the constraint
\eqref{MCMB1 rate}, we apply the SCA framework by
quadratically convexifying $-\log\det\left(\mathbf{I}_N
+\mathbf{H}_k\mathbf{V}_{g_k}\mathbf{V}_{g_k}^H\mathbf{H}_k^H
\mathbf{J}_k\right)$. Specifically,
given any fixed $\left\{\mathbf{V}_g^{(i)}\right\}$, we define a convex quadratic function:
\begin{eqnarray}\label{f}
h_k^{(i)}\left(\left\{\mathbf{V}_g\right\}\right)&\triangleq&
\sum_{g=1}^G\mathrm{Tr}\left(\mathbf{V}_{g}^H
\mathbf{\hat{A}}_{k}^{(i)}\mathbf{V}_{g}\right)
\nonumber\\
&&+2\Re\left\{\mathrm{Tr}\left(\mathbf{\hat{B}}_{k}^{(i)}\mathbf{V}_{g_k}\right)\right\}
+\hat{b}_{k}^{(i)},
\end{eqnarray}
where $\mathbf{\hat{A}}_{k}^{(i)}$, $\mathbf{\hat{B}}_{k}^{(i)}$, and $\hat{b}_{k}^{(i)}$ are
given by
\begin{eqnarray}
\mathbf{\hat{A}}_{k}^{(i)}&\triangleq&
\mathbf{H}_{k}^H\mathbf{\hat{U}}_{k}^{(i)}\left(\mathbf{I}_{d}-\left(\mathbf{\hat{U}}_{k}^{(i)}\right)^H
\mathbf{H}_k\mathbf{V}_{g_k}^{(i)}\right)^{-1}
\nonumber\\
&&\cdot\left(\mathbf{\hat{U}}_{k}^{(i)}\right)^H\mathbf{H}_k,
\label{D}\\
\mathbf{\hat{B}}_{k}^{(i)}&\triangleq&
-\left(\mathbf{I}_{d}-\left(\mathbf{\hat{U}}_{k}^{(i)}\right)^H
\mathbf{H}_{k}\mathbf{V}_{g_k}^{(i)}\right)^{-1}
\nonumber\\
&&\cdot
\left(\mathbf{\hat{U}}_{k}^{(i)}\right)^H\mathbf{H}_{k},
\label{C}\\
\hat{b}_{k}^{(i)}&\triangleq&
\mathrm{Tr}\Bigg(\left(\mathbf{I}_{d}-\left(\mathbf{\hat{U}}_{k}^{(i)}\right)^H
\mathbf{H}_k\mathbf{V}_{g_k}^{(i)}\right)^{-1}
\nonumber\\
&&\left(\mathbf{I}_{d}
+\sigma_k^2\left(\mathbf{\hat{U}}_{k}^{(i)}\right)^H\mathbf{\hat{U}}_{k}^{(i)}\right)\Bigg)
\nonumber\\
&&+\log\det\left(\mathbf{I}_{d}-\left(\mathbf{\hat{U}}_{k}^{(i)}\right)^H
\mathbf{H}_k\mathbf{V}_{g_k}^{(i)}\right)-d,~~~~~
\label{b}
\end{eqnarray}
with
\begin{equation}
\mathbf{\hat{U}}_{k}^{(i)}\triangleq
\left(\sum_{g=1}^G\mathbf{H}_k\mathbf{V}_{g}^{(i)}
\left(\mathbf{V}_{g}^{(i)}\right)^H\mathbf{H}_k^H
+\sigma_k^2\mathbf{I}_{N}\right)^{-1}\mathbf{H}_k\mathbf{V}_{g_k}^{(i)}.\label{Ui}
\end{equation}
Notice that $h_k^{(i)}\left(\left\{\mathbf{V}_g\right\}\right)$ is in a similar form to
$\phi_{k,t}^{(i)}\left(\left\{\mathbf{V}_{g,t}\right\}_{g=1}^G\right)$ in \eqref{phi},
thus by using similar arguments as in Appendix \ref{proof of lm2},
we can establish
two properties of $h_k^{(i)}\left(\left\{\mathbf{V}_g\right\}\right)$:
\begin{enumerate}
\item[\textbf{C.a}]
 $h_k^{(i)}\left(\left\{\mathbf{V}_g\right\}\right)\geq-\log\det\left(\mathbf{I}_N
+\mathbf{H}_k\mathbf{V}_{g_k}\mathbf{V}_{g_k}^H\mathbf{H}_k^H
\mathbf{J}_k\right)$,
where the equality holds at $\mathbf{V}_g=\mathbf{V}_g^{(i)}$, $\forall g=1,2,\ldots,G$.
\item[\textbf{C.b}]
 $\frac{\partial }{\partial v}h_k^{(i)}\left(\left\{\mathbf{V}_g^{(i)}\right\}\right)=-\frac{\partial}{\partial v}\Big(\log\det\Big(\mathbf{I}_N$
$+\mathbf{H}_k\mathbf{V}_{g_k}^{(i)}\left(\mathbf{V}_{g_k}^{(i)}\right)^H\mathbf{H}_k^H
\mathbf{J}_k^{(i)}\Big)\Big)$, where
$v$ represents any element of
$\mathbf{V}_g$, $\forall g=1,2,\ldots,G$, and
$\mathbf{J}_k^{(i)}=\left.\mathbf{J}_k\right\vert_{\big\{\mathbf{V}_g\big\}=\big\{\mathbf{V}_g^{(i)}\big\}}$.
\end{enumerate}

Consequently, the nonconvex constraint \eqref{MCMB1 rate} can be tightly approximated by
\begin{equation}\label{convex c1}
(F_{g_k}-C_k)\eta_{g_k}+h_k^{(i)}\left(\left\{\mathbf{V}_g\right\}\right)\leq 0,~~\forall k=1, 2, \ldots, K.
\end{equation}
Since $h_k^{(i)}\left(\left\{\mathbf{V}_g\right\}\right)$ in \eqref{f} is convex quadratic over
$\left\{\mathbf{V}_g\right\}$, and $(F_{g_k}-C_k)\eta_{g_k}$ is linear over $\eta_{g_k}$,
the constructed constraint in \eqref{convex c1} is jointly convex over $\left\{\mathbf{V}_g\right\}$ and $\eta_{g_k}$.
With the sequence of convex constraints constructed in \eqref{convex c1},
problem \eqref{MCMB0} can be iteratively solved in the SCA framework,
with the $i$-th SCA subproblem shown in \eqref{MCMB2}.

\section{Proposed Algorithm for Solving \eqref{MULTI}}\label{al5}

\begin{breakablealgorithm}
\caption{Proposed Algorithm for Solving \eqref{MULTI}}
\begin{algorithmic}[1]\footnotesize
\State Initialize $\left\{\mathbf{V}^{(0)}_{g,t,\mathbf{f}}, \eta^{(0)}_{g,t,\mathbf{f}}\right\}$
and $\left\{C^{(0)}_{k,f_{g_k}}\right\}$:
\begin{eqnarray}
&&C_{k,f_{g_k}}^{(0)}=C_{\text{tot}}/\left(K\left\vert\mathcal{F}_{g_k}\right\vert\right),
~~~~\mathbf{V}_{g,t,\mathbf{f}}^{(0)}=\sqrt{\frac{P_{\text{tot}}}{GMd}}\mathbf{1}_{M\times d},
\nonumber\\
&&\eta_{g,t,\mathbf{f}}^{(0)}=\min_{k\in\mathcal{K}_g}
\Bigg\{\frac{1}{F_{f_g}-C_{k,f_g}^{(0)}}
\nonumber\\
&&\cdot\log\det\left(\mathbf{I}_N
+\mathbf{H}_{k,t,\mathbf{f}}\mathbf{V}_{g,t,\mathbf{f}}^{(0)}
\left(\mathbf{V}_{g,t,\mathbf{f}}^{(0)}\right)^H\mathbf{H}_{k,t,\mathbf{f}}^H
\mathbf{J}_{k,t,\mathbf{f}}^{(0)}\right)\Bigg\}.
~~~~~~~~~~~~~~~~~~~~~~~~~~~~~~~~~~~~~~~~~~~~~~~~~~~~~~~~~~~~~~~~~~~
~~~~~~~~~~~~~~~~~~~~~~~~~~~~~\nonumber
\end{eqnarray}\\
$\textbf{repeat}$ ($i=0,1,\ldots$)\\
Compute $\left\{
\mathbf{U}_{k,t,\mathbf{f}}^{(i)},
\mathbf{A}_{k,t,\mathbf{f}}^{(i)},
\mathbf{B}_{k,t,\mathbf{f}}^{(i)},
b_{k,t,\mathbf{f}}^{(i)}
\right\}$:
\begin{eqnarray}
&\mathbf{U}_{k,t,\mathbf{f}}^{(i)}=&
\left(\sum_{g=1}^G\mathbf{H}_{k,t,\mathbf{f}}\mathbf{V}_{g,t,\mathbf{f}}^{(i)}
\left(\mathbf{V}_{g,t,\mathbf{f}}^{(i)}\right)^H\mathbf{H}_{k,t,\mathbf{f}}^H
+\sigma_k^2\mathbf{I}_{N}\right)^{-1}
\nonumber\\
&&
\mathbf{H}_{k,t,\mathbf{f}}\mathbf{V}_{g_k,t,\mathbf{f}}^{(i)},
~~~~~~~~~~~~~~~~~~~~~~~~~~~~~~~~~~~~~~~~~~~~~~~~~~~\nonumber\\
&\mathbf{A}_{k,t,\mathbf{f}}^{(i)}=&
\mathbf{H}_{k,t,\mathbf{f}}^H\mathbf{U}_{k,t,\mathbf{f}}^{(i)}\left(\mathbf{I}_{d}-\left(\mathbf{U}_{k,t,\mathbf{f}}^{(i)}\right)^H
\mathbf{H}_{k,t,\mathbf{f}}\mathbf{V}_{g_k,t,\mathbf{f}}^{(i)}\right)^{-1}
\nonumber\\
&&
\cdot\left(\mathbf{U}_{k,t,\mathbf{f}}^{(i)}\right)^H\mathbf{H}_{k,t,\mathbf{f}},
~~~~~~~~~~~~~~~~~~~~~~~~~~~~~~~~~~~~~~~~~~~~~~~~~~~~~\nonumber\\
&\mathbf{B}_{k,t,\mathbf{f}}^{(i)}=&
-\left(\mathbf{I}_{d}-\left(\mathbf{U}_{k,t,\mathbf{f}}^{(i)}\right)^H
\mathbf{H}_{k,t,\mathbf{f}}\mathbf{V}_{g_k,t,\mathbf{f}}^{(i)}\right)^{-1}
\nonumber\\
&&
\cdot\left(\mathbf{U}_{k,t,\mathbf{f}}^{(i)}\right)^H\mathbf{H}_{k,t,\mathbf{f}},\nonumber\\
&b_{k,t,\mathbf{f}}^{(i)}=&\mathrm{Tr}\Bigg(\left(\mathbf{I}_{d}
-\left(\mathbf{U}_{k,t,\mathbf{f}}^{(i)}\right)^H
\mathbf{H}_{k,t,\mathbf{f}}\mathbf{V}_{g_k,t,\mathbf{f}}^{(i)}\right)^{-1}
\nonumber\\
&&\cdot
\left(\mathbf{I}_{d}
+\sigma_k^2\left(\mathbf{U}_{k,t,\mathbf{f}}^{(i)}\right)^H\mathbf{U}_{k,t,\mathbf{f}}^{(i)}\right)\Bigg)
~~~~~~~~~~~~~~~~~~~~~~~~~~~~~~~~~~~~~~~~~~~~~~~~~~~~~~~\nonumber\\
&&+\log\det\left(\mathbf{I}_d-\left(\mathbf{U}_{k,t,\mathbf{f}}^{(i)}\right)^H
\mathbf{H}_{k,t,\mathbf{f}}\mathbf{V}_{g_k,t,\mathbf{f}}^{(i)}\right)
\nonumber\\
&&
+\frac{\left(\eta_{g_k,t,\mathbf{f}}^{(i)}+C_{k,f_{g_k}}^{(i)}\right)^2}{2}-d.\nonumber
\end{eqnarray}\\
Initialize $\theta^{(0)}=1, \mu=1,
\delta_{t,\mathbf{f}}=1,
\lambda_{k,t,\mathbf{f}}=1, \forall \mathbf{f}\in\mathcal{F}, \forall k=1,2,\ldots,K, \forall t=1,2,\ldots,T$.\\
$\textbf{repeat}$ ($s=1,2,\ldots$)\\
$\theta^{(s)}=\frac{1+\sqrt{1+4\left(\theta^{(s-1)}\right)^2}}{2}$.\\
Update $\left\{\mathbf{V}^{\diamond}_{g,t,\mathbf{f}}, \eta^{\diamond}_{g,t,\mathbf{f}}\right\}$
and $\left\{C^{\diamond}_{k,f_{g_k}}\right\}$:
\begin{eqnarray}
\mathbf{V}_{g,t,\mathbf{f}}^{\diamond}&=&\left(\left(\rho_2+\delta_{t,\mathbf{f}}\right)\mathbf{I}_M
+\sum_{k=1}^K\lambda_{k,t,\mathbf{f}}\mathbf{A}_{k,t,\mathbf{f}}^{(i)}\right)^{-1}
\nonumber\\
&&
\cdot\left(\rho_2\mathbf{V}_{g,t,\mathbf{f}}^{(i)}-\sum_{k\in\mathcal{K}_g}\lambda_{k,t,\mathbf{f}}
\left(\mathbf{B}_{k,t,\mathbf{f}}^{(i)}
\right)^H\right),
\nonumber\\
\eta_{g,t,\mathbf{f}}^{\diamond}&=&
\frac{\left(p_{f_g}-\sum_{k\in\mathcal{K}_g}\lambda_{k,t,\mathbf{f}}\right)F_{f_g}+
\sum_{k\in\mathcal{K}_g}\lambda_{k,t,\mathbf{f}}C_{k,f_{g_k}}^{(i)}}
{\rho_1+\sum_{k\in\mathcal{K}_g}\lambda_{k,t,\mathbf{f}}}
\nonumber\\
&&+\eta_{g,t,\mathbf{f}}^{(i)},
\nonumber\\
C_{k,f_{g_k}}^{\diamond}&=&\min\Bigg\{\max\Bigg\{\Bigg(C_{k,f_{g_k}}^{(i)}
\nonumber\\
&&
+\frac{\sum_{t=1}^T\sum_{f_{g_k}}\lambda_{k,t,\mathbf{f}}\eta_{g_k,t,\mathbf{f}}^{(i)}-\mu}
{\rho_3+\sum_{t=1}^T\sum_{f_{g_k}}\lambda_{k,t,\mathbf{f}}}\Bigg), 0\Bigg\}, F_{f_{g_k}}\Bigg\}.
~~~~~~~~~~~~~~~~~~~~~~~~~~~~~~~~~~~~~~~~~~~~~~~~~~~~~~~~~~~~~~~~~~~~~~~~~~~~~~~~~~~~~~~~\nonumber
\end{eqnarray}\\
Update $\left\{\tilde{\delta}_{t,\mathbf{f}}^{(s)}\right\}$, $\left\{\tilde{\lambda}_{k,t,\mathbf{f}}^{(s)}\right\}$, and $\tilde{\mu}^{(s)}$:
\begin{eqnarray}
\tilde{\delta}_{t,\mathbf{f}}^{(s)}&=&\left(\delta_{t,\mathbf{f}}+\beta_s \left(\sum_{g=1}^G\mathrm{Tr}\left(\mathbf{V}^{\diamond}_{g,t,\mathbf{f}}
\left(\mathbf{V}^{\diamond}_{g,t,\mathbf{f}}\right)^H\right)-P_{\text{tot}}\right)\right)^{+},~~~~~~~~~
\nonumber\\
\tilde{\lambda}_{k,t,\mathbf{f}}^{(s)}&=&\left(\lambda_{k,t,\mathbf{f}}+\beta_s f_{k,t,\mathbf{f}}^{(i)}\left(C_{k,f_{g_k}}^{\diamond}, \eta_{g_k,t,\mathbf{f}}^{\diamond}, \left\{\mathbf{V}_{g,t,\mathbf{f}}^{\diamond}\right\}_{g=1}^G\right)\right)^{+},~~~~~~~~~~~~
\nonumber\\
\tilde{\mu}^{(s)}&=&\left(\mu+\beta_s \left(\sum_{k=1}^K\sum_{f_{g_k}\in\mathcal{F}_{g_k}}C_{k,f_{g_k}}^{\diamond}-C_{\text{tot}}\right)\right)^{+}.
~~~~~~~~~~~~~~~~~~~~~~~~~~~~~~~~~~~~~~~~~~~~~~~~~~~~~~~~~~~~~~~~~~~~~~~~~~~~~~~~~~~~~~~~~~~~~\nonumber
\end{eqnarray}\\
Update $\left\{\delta_{t,\mathbf{f}}\right\}$, $\left\{\lambda_{k,t,\mathbf{f}}\right\}$, and $\mu$:
\begin{eqnarray}
&&\delta_{t,\mathbf{f}}=\tilde{\delta}_{t,\mathbf{f}}^{(s)}
+\frac{\theta^{(s-1)}-1}{\theta^{(s)}}\left(\tilde{\delta}_{t,\mathbf{f}}^{(s)}
-\tilde{\delta}_{t,\mathbf{f}}^{(s-1)}\right),
\nonumber\\
&&\lambda_{k,t,\mathbf{f}}=\tilde{\lambda}_{k,t,\mathbf{f}}^{(s)}
+\frac{\theta^{(s-1)}-1}{\theta^{(s)}}
\left(\tilde{\lambda}_{k,t,\mathbf{f}}^{(s)}-\tilde{\lambda}_{k,t,\mathbf{f}}^{(s-1)}\right),
~~~~
\nonumber\\
&&\mu=\tilde{\mu}^{(s)}+\frac{\theta^{(s-1)}-1}{\theta^{(s)}}
\left(\tilde{\mu}^{(s)}-\tilde{\mu}^{(s-1)}\right).
~~~~~~~~~~~~~~~~~~~~~~~~~~~~~~~~~~~~~~~~~~\nonumber
\end{eqnarray}\\
$\textbf{until}$ convergence\\
$\left\{\mathbf{V}^{(i+1)}_{g,t,\mathbf{f}}, \eta^{(i+1)}_{g,t,\mathbf{f}}\right\}=\left\{\mathbf{V}^{\diamond}_{g,t,\mathbf{f}}, \eta^{\diamond}_{g,t,\mathbf{f}}\right\}$
and $\left\{C^{(i+1)}_{k,f_{g_k}}\right\}=\left\{C^{\diamond}_{k,f_{g_k}}\right\}$.\\
$\textbf{until}$ convergence
\end{algorithmic}
\end{breakablealgorithm}

% \section*{Acknowledgment}

\bibliographystyle{IEEEtran}
\bibliography{IEEEabrv,LiYang}

\begin{IEEEbiography}[{\includegraphics[width=1in,height=1.25in,clip,keepaspectratio]{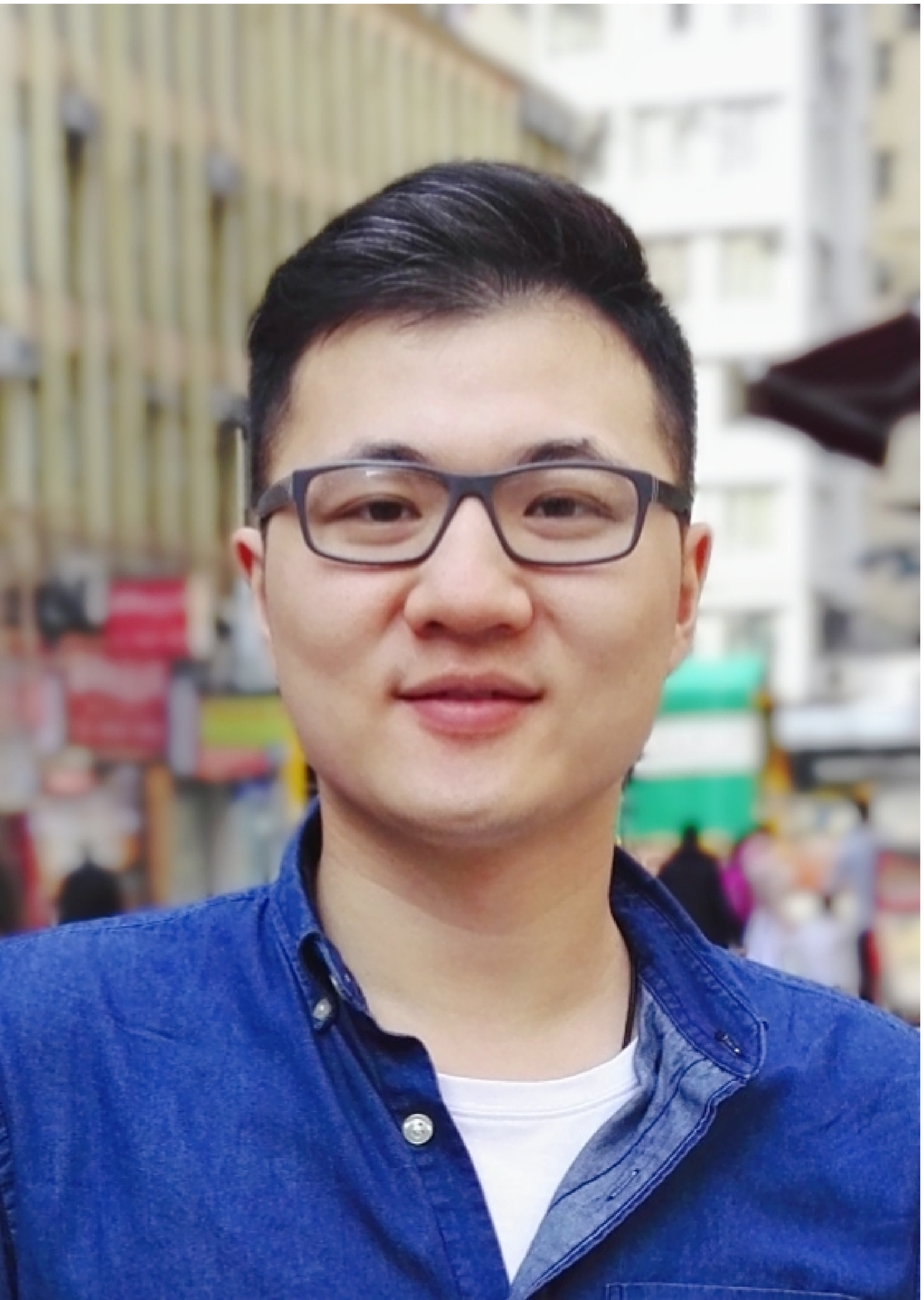}}]{Yang Li} (S'14) received the B.E degree and the M.E degree
in electronics engineering from Beihang University
(BUAA), Beijing, China, in 2012 and 2015, respectively. He received the Ph.D. degree
in electrical and electronic engineering at the University of Hong Kong (HKU),
Hong Kong, China, in 2019.
His research interests are in general areas of wireless communications and signal processing.
\end{IEEEbiography}

\begin{IEEEbiography}[{\includegraphics[width=1in, height=1.25in, clip, keepaspectratio]{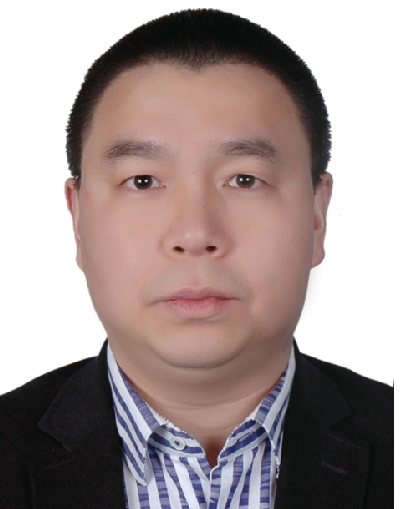}}]{Minghua Xia} (M'12) received his Ph.D. degree in Telecommunications and Information Systems from Sun Yat-sen University, Guangzhou, China, in 2007. Since 2015, he has been a Professor with Sun Yat-sen University.

From 2007 to 2009, he was with the Electronics and Telecommunications Research Institute (ETRI) of South Korea, Beijing R\&D Center, Beijing, China, where he worked as a member and then as a senior member of engineering staff. From 2010 to 2014, he was in sequence with The University of Hong Kong, Hong Kong, China; King Abdullah University of Science and Technology, Jeddah, Saudi Arabia; and the Institut National de la Recherche Scientifique (INRS), University of Quebec, Montreal, Canada, as a Postdoctoral Fellow. His research interests are in the general areas of wireless communications and signal processing.

Dr. Xia received the Professional Award at the IEEE TENCON, held in Macau, in 2015. He was recognized as Exemplary Reviewer by {\scshape IEEE Transactions on Communications} in 2014, {\scshape IEEE Communications Letters} in 2014, and {\scshape IEEE Wireless Communications Letters} in 2014 and 2015. Dr. Xia serverd as TPC Symposium Chair of IEEE ICC'2019 and now serves as Associate Editor for the {\scshape IEEE Transactions on Cognitive Communications and Networking} and the {\scshape IET Smart Cities}.
\end{IEEEbiography}

\begin{IEEEbiography}[{\includegraphics[width=1in,height=1.25in,clip,keepaspectratio]{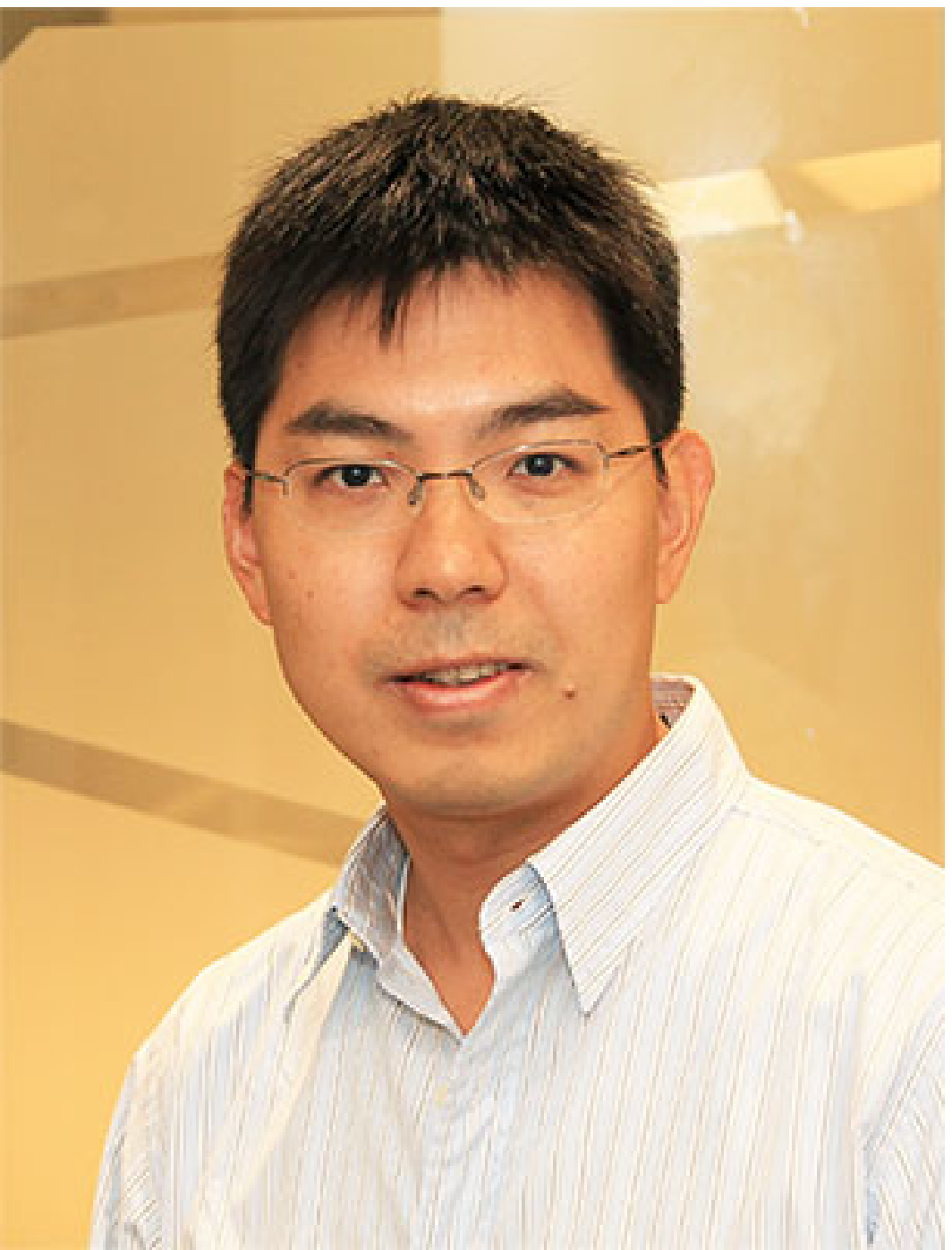}}]{Yik-Chung Wu} received the B.Eng. (EEE) degree in 1998 and the M.Phil. degree in 2001 from the University of Hong Kong (HKU). He received the Croucher Foundation scholarship in 2002 to study Ph.D. degree at Texas A\&M University, College Station, and graduated in 2005. From August 2005 to August 2006, he was with the Thomson Corporate Research, Princeton, NJ, as a Member of Technical Staff.  Since September 2006, he has been with HKU, currently as an Associate Professor. He was a visiting scholar at Princeton University, in summers of 2015 and 2017. His research interests are in general areas of signal processing, machine learning and communication systems. Dr. Wu served as an Editor for {\scshape IEEE Communications Letters} and {\scshape IEEE Transactions on Communications}, and is currently an editor for {\scshape Journal of Communications and Networks}.
\end{IEEEbiography}

\end{document}